\documentclass[lettersize,journal]{IEEEtran}
\IEEEoverridecommandlockouts

\usepackage{cite}
\usepackage{bm} 
\usepackage{amsmath,amssymb,amsfonts,amsthm}
\renewcommand{\IEEEQED}{\IEEEQEDopen}

\newtheoremstyle{indentedplain}%
{\topsep}
{\topsep}
{\itshape}
{0.5em}
{\bfseries}
{.}
{0.5em}
{}
\theoremstyle{indentedplain}
\newtheorem{theorem}{Theorem}
\newtheorem{proposition}{Proposition}
\newtheorem{lemma}{Lemma}
\newtheorem{corollary}{Corollary}
\newtheorem{assumption}{Assumption}
\newtheoremstyle{indenteddefinition}%
{\topsep}
{\topsep}
{\normalfont}
{0.5em}
{\bfseries}
{.}
{0.5em}
{}
\theoremstyle{indenteddefinition}
\newtheorem{definition}{Definition}
\newtheoremstyle{boldremark}%
{3pt}
{3pt}
{\normalfont}
{0.5em}
{\bfseries}
{.}
{0.5em}
{}
\theoremstyle{boldremark}
\newtheorem{remark}{Remark}
\usepackage{algorithmic}
\usepackage{array}
\usepackage{booktabs}
\usepackage{graphicx}
\usepackage{flushend}
\usepackage{textcomp}
\usepackage{hyperref}					
\hypersetup{colorlinks,
	linkcolor=blue,%
	anchorcolor=blue,
    citecolor=blue}
\usepackage{xcolor}

\newif\ifarxivversion
\arxivversiontrue
\newcommand{\apprefappendix}[2]{\ifarxivversion Appendix~\ref{#1}\else arXiv Appendix~#2\fi}

\def\BibTeX{{\rm B\kern-.05em{\sc i\kern-.025em b}\kern-.08em
    T\kern-.1667em\lower.7ex\hbox{E}\kern-.125emX}}
\begin{document}

\title{{ How Much Sensing Information Is Needed to Control an Unstable Linear System?}
}
\author{
Ming Li,~\IEEEmembership{Student Member,~IEEE},
Fan Liu,~\IEEEmembership{Senior Member,~IEEE},
Yifeng Xiong,~\IEEEmembership{Member,~IEEE},\\
Jie Xu,~\IEEEmembership{Fellow,~IEEE},
Tao Liu,~\IEEEmembership{Member,~IEEE},
Weijie Yuan,~\IEEEmembership{Senior Member,~IEEE},
and Shi Jin,~\IEEEmembership{Fellow,~IEEE}%
\thanks{This work was supported in part by the National Science and Technology Major Projects of China under Grant 2025ZD1302000, and the National Natural Science Foundation of China (NSFC) under Grant 62522107. (\textit{Corresponding author: Fan Liu.})}
\thanks{M. Li, T. Liu and W. Yuan are with the School of Automation and Intelligent
Manufacturing, Southern University of Science and Technology, Shenzhen 518055, China
(email: lim2024@mail.sustech.edu.cn, liut6@sustech.edu.cn, yuanwj@sustech.edu.cn).}%
\thanks{F. Liu and S. Jin are with the National Mobile Communications Research Laboratory, Southeast University, Nanjing 210096, China  (email: fan.liu@seu.edu.cn,
jinshi@seu.edu.cn).}%
\thanks{Y. Xiong is with the School of Information and Communication Engineering, Beijing University of Posts and Telecommunications, Beijing 100876, China
(email: yifengxiong@bupt.edu.cn).}%
\thanks{J. Xu is with the School of Science and Engineering, The Chinese University of Hong Kong, Shenzhen 518172, China
(email: xujie@cuhk.edu.cn).}%
}

\maketitle

\begin{abstract}
Modern control systems increasingly rely on sensing to infer the system state before control actions can be taken. Yet a given observation mechanism may fail to preserve sufficient information about the unstable modes, regardless of the downstream estimator or controller. This paper asks how much sensing information is needed to estimate and control an unstable linear system, whose measurements are generated by a prescribed, possibly nonlinear and non-Gaussian, observation law $p(\mathbf y_t|\mathbf x_t)$.
To address this question, we first quantify sensing information using
directed information, thereby accounting for causal feedback.
We then establish necessary and sufficient information rate conditions for
estimating and controlling this linear system. For necessity, keeping either the estimation error or the closed-loop state
bounded in mean square requires a directed information rate of at least the
open-loop expansion rate $R_{\mathrm{exp}}$.
This lower bound remains valid under additive process noise.
Since the directed information rate from the unstable state is generally difficult to evaluate, we derive computable upper and lower bounds when each observation is generated by a nonlinear function of the state plus additive noise. An upper bound below $R_{\mathrm{exp}}$ certifies infeasibility, whereas a lower bound above $R_{\mathrm{exp}}+R_{\mathrm{NG}}$ certifies sufficiency under the stated posterior covariance regularity.
For linear Gaussian observations, the upper bound is tight and its value is determined by the steady-state Riccati equation.
For sufficiency, we introduce the posterior non-Gaussianity rate
$R_{\mathrm{NG}}$, which {quantifies the additional information rate required to bridge the gap between the posterior entropy and its covariance-matched Gaussian}.
Under uniform posterior covariance regularity, the condition
$\mathcal I(\mathbf z\to\mathbf y)>
R_{\mathrm{exp}}+R_{\mathrm{NG}}$
guarantees that the estimation error converges to zero in mean square.
For a stabilizable plant, certainty-equivalence feedback further guarantees mean-square convergence of the
closed-loop state to zero.
Finally, we identify verifiable curvature conditions on the likelihood and prior under which $R_{\mathrm{NG}}=0$, so that the sufficient information threshold coincides with $R_{\mathrm{exp}}$. Together, these results clarify how unstable dynamics and posterior non-Gaussianity jointly determine how much sensing information control requires.
\end{abstract}


\IEEEraisesectionheading{\section{Introduction}}
\IEEEPARstart{F}{eedback} control begins with sensing. Before an autonomous or robotic system can take effective control actions, it must first infer its state from the available measurements. When these measurements are corrupted by noise, saturation, or nonlinear distortion, however, information about unstable modes may be lost before estimation and control even begin. This raises a basic question: how much information must a sensing mechanism provide to make estimation and stabilization possible? Toward that end, we study an unstable linear system observed through a prescribed conditional law $p(\mathbf y_t|\mathbf x_t)$, and characterize its sensing information requirements using tools from both control theory and information theory.

The classical data-rate theorem provides a natural starting point for this question. It considers an unstable system whose state has already been measured and is transmitted to the controller over a finite-rate communication channel. Its central conclusion is that the information delivered through the channel must keep pace with the open-loop expansion of the unstable modes\cite{wong1999finite2,brockett2000quantized,elia2001limited,
liberzon2003stabilization,tatikonda2004control,
tatikonda2004controlnoisy}. This principle relating information rate to system expansion has since been extended to stochastic disturbances and
fixed-rate feedback \cite{Girish2004rates,yuksel2010fixedrate}, noisy
channels \cite{tatikonda2004stochastic}, and signal-to-noise-ratio
constraints \cite{braslavsky2007feedback}.

The same relationship between information rate and system expansion guides our analysis, but the information constraint arises at a fundamentally different stage. In the classical setting, the state is assumed to have been measured perfectly, and information is lost only when that measurement is encoded and transmitted over a communication link. Here, the bottleneck appears before estimation: the observation law determines how much information about the state is contained in the measurements available to the estimator and controller, as illustrated in Fig.~\ref{fig:comparison}. We refer to control under this information constraint at the observation side as \emph{sensing-limited control}.

Existing studies on control under sensing constraints have largely treated the sensing mechanism as part of the system design. Typical examples include sensor selection and switching among sensing modes within LQG formulations \cite{tzoumas2021lqg,vanhorssen2019switched}. Related information-theoretic work minimizes directed information subject to an LQG performance constraint and shows that the optimum can be realized by a linear Gaussian sensor, filter, and controller \cite{tanaka2018minimumDI}. Our problem is complementary: rather than optimizing the sensing architecture, we evaluate whether the observations generated by a given mechanism convey enough causal information about the unstable modes for estimation and stabilization. To the best of our knowledge, the required information rate has not been characterized for this setting. This motivates a sensing counterpart to the data-rate theorem.

To extend this principle to the observation side, we first need an information measure for a given observation law. Unlike a finite-rate communication channel, an observation law does not come with an explicit rate parameter. Although mutual information measures statistical dependence between the state and the observations, it is symmetric and therefore does not distinguish causal information flow from dependence induced by feedback. Directed information from the unstable state process to the observation process provides a natural causal measure \cite{Massey1990CAUSALITYFA,Charalambous2013}. At each time, its increment measures the new information carried by the current observation about the unstable state history, conditioned on all past observations. Its rate therefore quantifies the causal information supplied by the observation process.

Under the above formulation, the central question may thus be divided into two parts. First, does the open-loop expansion rate remain the fundamental information threshold when the measurements are generated by a given nonlinear observation law? Answering this question would reveal whether the threshold is imposed solely by the unstable dynamics or also depends on the particular form of the observation mechanism. Second, if the directed information rate exceeds this threshold, is this condition alone sufficient to drive the estimation error and closed-loop state to zero in mean square? This question is less immediate because directed information characterizes posterior uncertainty through entropy, whereas mean square performance depends on posterior covariance. Although these quantities are directly related for Gaussian posteriors, they may behave differently when the posterior is non-Gaussian.

\begin{figure}[!t]
    \centering
    \includegraphics[width=\linewidth]{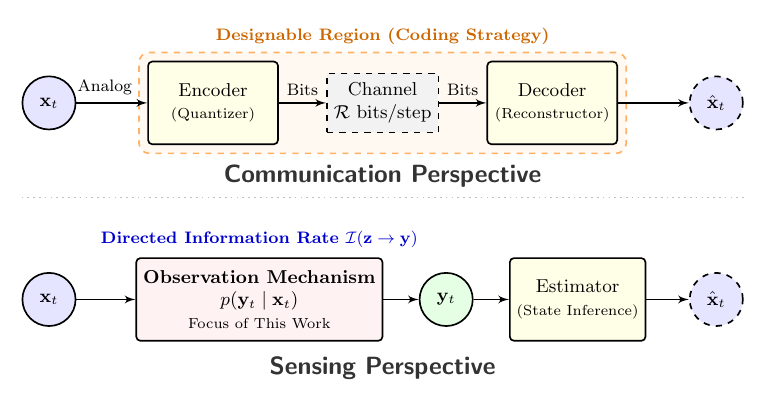}
    \caption{Control under communication constraints uses a designable interface between the encoder and channel, whereas sensing-limited control evaluates a given observation mechanism.}
    \label{fig:comparison}
\end{figure}

Building on our preliminary conference paper~\cite{li2026sensing}, the present work develops a more comprehensive information-theoretic framework for sensing-limited control and substantially extends both the converse and sufficiency analyses. The main contributions are summarized as follows:

\begin{itemize}
\item We derive a fundamental lower bound on the sensing information required for mean-square observability and stabilizability under a given observation law that may be nonlinear and non-Gaussian. Specifically, the directed information rate from the unstable state to the observations must satisfy
\begin{equation*}
\mathcal{I}(\mathbf z\to\mathbf y)
\geq
R_{\mathrm{exp}}.
\end{equation*}
The proof is based on a conditional entropy balance. We further show that the same lower bound remains necessary in the presence of additive process disturbances.
\item We derive a computable impossibility criterion for nonlinear observation models with additive noise. By deriving an upper bound on the directed information rate from the full state in terms of the conditional output covariance, we obtain a criterion that certifies infeasibility whenever this upper bound falls below $R_{\mathrm{exp}}$. For the linear Gaussian model, the upper bound is attained, and the exact directed information rate from the full state is characterized by the steady-state solution of the discrete algebraic Riccati equation.
\item We establish a general sufficient information threshold that accounts for posterior non-Gaussianity through the rate $R_{\mathrm{NG}}$. Under uniform posterior covariance regularity, the condition
\begin{equation*}
\mathcal{I}(\mathbf z\to\mathbf y)
>R_{\mathrm{exp}}+R_{\mathrm{NG}}
\end{equation*}
guarantees mean-square attractivity of the estimation error and, under certainty-equivalence feedback, of the closed-loop state. The proof bridges the gap between entropy and covariance by converting posterior entropy contraction into covariance convergence.  With an additional local mean-square overshoot bound, these attractivity results further imply asymptotic mean-square observability and stabilizability.
\item When each observation is generated by a nonlinear function of the state plus additive noise, we further derive an
lower bound based on entropy power that provides a computable way to verify the
sufficient rate condition without evaluating the directed information
rate exactly. We also identify important model classes for which
$R_{\mathrm{NG}}=0$, including linear Gaussian observations and their invertible readout variants, and provide a verifiable curvature condition
based on eventual posterior log-concavity.
\end{itemize}

The remainder of this paper is organized as follows.
Section~\ref{sec:system-model} presents the system model, the stable and unstable
decomposition, and the information and performance measures.
Section~\ref{sec:necessary} establishes the necessary conditions, their
extension to process noise, the computable impossibility criterion, and its
exact linear Gaussian specialization.
Section~\ref{sec:sufficient} develops the general sufficiency result, a
lower bound based on entropy power for verifying its rate condition, and the
regimes where $R_{\mathrm{NG}}=0$.
Section~\ref{sec:conclusion} concludes the paper.

\textit{Notations:}
Matrices are denoted by bold uppercase letters, vectors by bold lowercase
letters, and scalars by normal italic letters. Random quantities and their
realizations use the same typography; their roles are distinguished by
context, probability laws, and expectation operators. The superscript
$(\cdot)^T$ denotes transpose; $\mathbf I_n$ is the $n\times n$ identity
matrix; and $\mathrm{Tr}(\cdot)$ and $\det(\cdot)$ denote trace
and determinant. The relation $\mathbf A\succeq\mathbf B$ means
that $\mathbf A-\mathbf B$ is positive semidefinite.
The norms $\|\mathbf x\|$ and $\|\mathbf A\|$ are Euclidean and
spectral, respectively. For a symmetric matrix $\mathbf A$,
$\lambda_{\max}(\mathbf A)$ and $\lambda_{\min}(\mathbf A)$ are
its largest and smallest eigenvalues; $\sigma_{\max}(\mathbf A)$ and
$\sigma_{\min}(\mathbf A)$ are its largest and smallest singular
values.
Let $\mathbb N_0:=\{0,1,2,\ldots\}$ and, for integers $j\le k$, let
$[j\!:\!k]:=\{j,j+1,\ldots,k\}$; $\lfloor a\rfloor$ is the greatest
integer not exceeding $a$. For a sequence $\{\mathbf a_t\}$ indexed by time,
define its history by
$\mathbf a^t:=(\mathbf a_0,\ldots,\mathbf a_t)$.
A history ending at $-1$ is empty; in particular,
$\mathbf y^{-1}:=\varnothing$.
The operators $\mathbb P(\cdot)$, $\mathbb E[\cdot]$,
$\mathrm{Var}(\cdot)$, and $\mathrm{Cov}(\cdot)$ denote probability,
expectation, variance, and covariance. The distribution
$\mathcal N(\boldsymbol\mu,\boldsymbol\Sigma)$ is Gaussian with mean
$\boldsymbol\mu$ and covariance $\boldsymbol\Sigma$. Conditional
quantities without a specified realization, such as
$h(\mathbf x|\mathbf y)$, are averaged over the random conditioning
variable. For a fixed realization $\mathbf y$, the conditional density is
written as $p_{\mathbf x|\mathbf y}(\mathbf x|\mathbf y)$; the same symbols
are used when the arguments are random, with the interpretation clear from
context.
For a random vector $\mathbf x$ with density $p_{\mathbf x}$, define
$h(\mathbf x)
:=
h(p_{\mathbf x})
:=
-\int p_{\mathbf x}(\boldsymbol\xi)
\log_2 p_{\mathbf x}(\boldsymbol\xi)
\,\mathrm d\boldsymbol\xi$.
The conditional differential entropy is
$h(\mathbf x|\mathbf y):=\mathbb E_{\mathbf y}
[h(p_{\mathbf x|\mathbf y}(\cdot|\mathbf y))]$; hence, all
differential entropies are measured in bits. For probability densities
$p$ and $q$,
$D_{\mathrm{KL}}(p\|q):=\int p(\boldsymbol\xi)
\log_2(p(\boldsymbol\xi)/q(\boldsymbol\xi))\,\mathrm d\boldsymbol\xi$
is the Kullback--Leibler divergence in bits.

\section{System Model and Preliminaries} \label{sec:system-model}

\subsection{System Dynamics and Observations}
Consider a discrete-time linear dynamical system. The overall system architecture is depicted in Fig.~\ref{fig:system_model}. The system dynamics are governed by the following noiseless state equation
\begin{equation}
    \mathbf x_{t+1} = \mathbf A \mathbf x_t + \mathbf B \mathbf u_t, \quad t \in \mathbb{N}_0,
    \label{eq:system_dynamics}
\end{equation}
where $\mathbf x_t \in \mathbb{R}^n$ is the state vector, $\mathbf u_t \in \mathbb{R}^m$ is the control input vector, $\mathbf A \in \mathbb{R}^{n \times n}$ is the system matrix, and $\mathbf B \in \mathbb{R}^{n \times m}$ is the control matrix. We assume that the system matrix $\mathbf A$ is unstable (i.e., it has at least one eigenvalue with magnitude at least 1) and that the initial state $\mathbf x_0$ admits a valid probability density function (PDF) with a finite second moment and a finite differential entropy, i.e., $-\infty < h(\mathbf x_0) < \infty$.

The system is observed through a memoryless vector channel subject to noise and nonlinear distortions, yielding the observation $\mathbf y_t \in \mathbb{R}^p$. To capture these stochastic and nonlinear characteristics, the observation process is modeled by a general conditional PDF $p(\mathbf y_t|\mathbf x_t)$. The estimator and controller are designed downstream of this law and operate on the generated observation history $\mathbf y^t=(\mathbf y_0,\ldots,\mathbf y_t)$.


\begin{figure}[t] 
    \centering
    \includegraphics[width=0.9\linewidth]{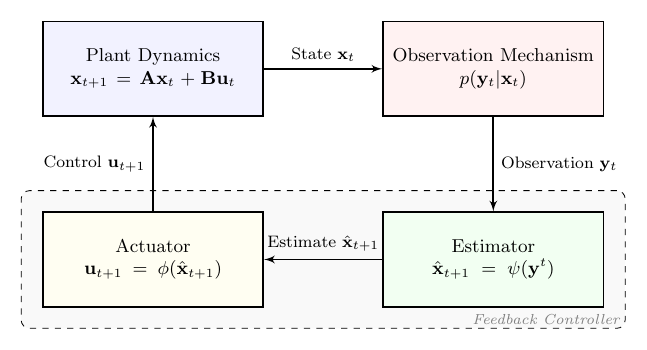}

    \caption{The closed-loop control architecture under a given observation mechanism. The feedback loop comprises the plant dynamics, the observation mechanism, an estimator, and an actuator.}
    \label{fig:system_model}
\end{figure}

\subsection{Canonical Decomposition and Expansion Rate} \label{subsec:canonical_decomposition}
To isolate the components critical to the information analysis, we decouple the system into stable and unstable modes.
Using a real Jordan decomposition, there exists an invertible transformation
matrix $\mathbf T\in\mathbb{R}^{n\times n}$ such that
\begin{equation*}
    \mathbf T\mathbf x_t
    =
    \begin{bmatrix}
        \mathbf z_t\\
        \mathbf s_t
    \end{bmatrix},
\end{equation*}
where $\mathbf z_t\in\mathbb R^{n_u}$ and
$\mathbf s_t\in\mathbb R^{n_s}$ denote the unstable and stable states,
respectively, with $n_u+n_s=n$.

Accordingly, the system matrices are transformed into the following block forms
\begin{equation*}
    \mathbf T\mathbf A\mathbf T^{-1}
    = \begin{bmatrix} \mathbf A_u & \mathbf 0 \\ \mathbf 0 & \mathbf A_s \end{bmatrix},
    \qquad
    \mathbf T\mathbf B
    = \begin{bmatrix} \mathbf B_u \\ \mathbf B_s \end{bmatrix}.
\end{equation*}
Here, $\mathbf A_u \in \mathbb{R}^{n_u \times n_u}$ contains all eigenvalues of $\mathbf A$ with magnitude $|\lambda| \ge 1$, while $\mathbf A_s \in \mathbb{R}^{n_s \times n_s}$ contains all eigenvalues with $|\lambda| < 1$. Similarly, $\mathbf B_u \in \mathbb{R}^{n_u \times m}$ and $\mathbf B_s\in \mathbb{R}^{n_s \times m}$ are the components of the control matrix corresponding to the unstable and stable subspaces.

The two subsystems therefore evolve as
\begin{equation*}
    \mathbf z_{t+1}=\mathbf A_u\mathbf z_t+\mathbf B_u\mathbf u_t,
    \qquad
    \mathbf s_{t+1}=\mathbf A_s\mathbf s_t+\mathbf B_s\mathbf u_t.
\end{equation*}

For each $t$, define the posterior and prior means and error covariances of
the unstable state as
\begin{equation*}
    \hat{\mathbf z}_t
    :=
    \mathbb E[\mathbf z_t|\mathbf y^{t}],
    \qquad
    \hat{\mathbf z}_t^-
    :=
    \mathbb E[\mathbf z_t|\mathbf y^{t-1}],
\end{equation*}
\begin{equation*}
    \mathbf P_t
    :=
    \mathrm{Cov}(\mathbf z_t|\mathbf y^{t}),
    \qquad
    \mathbf P_t^-
    :=
    \mathrm{Cov}(\mathbf z_t|\mathbf y^{t-1}).
\end{equation*}
These quantities are random through their dependence on the corresponding
observation histories.

Throughout, the estimation and control strategies are causal. Under the
predictive convention used in the converse arguments, the control input
$\mathbf u_t$ is generated through the composition of an estimator $\psi$
and a controller $\phi$
\begin{align*}
    \psi &: \mathcal Y^{t-1} \to \mathbb{R}^{n_x}, \quad \hat{\mathbf x}_t = \psi(\mathbf y^{t-1}), \\
    \phi &: \mathbb{R}^{n_x} \to \mathbb{R}^{n_u}, \quad \mathbf u_t = \phi(\hat{\mathbf x}_t).
\end{align*}
Here, $\mathcal Y_k$ denotes the observation space at time $k$, and
$\mathcal Y^{t-1}:=\prod_{k=0}^{t-1}\mathcal Y_k$ is the space of all
possible observation histories up to time $t-1$. Thus,
$\hat{\mathbf x}_t$ depends only on observations available before time
$t$. The function $\psi$ maps this history to a state estimate, and $\phi$
maps the estimate to the control input.


The information-theoretic limit discussed in the following sections is fundamentally governed by the expansion rate $R_{\mathrm{exp}}$, defined by
\begin{equation*}
    R_{\mathrm{exp}} := \sum_{i: |\lambda_i(\mathbf A)| \ge 1} \log_2 |\lambda_i(\mathbf A)| = \log_2 |\det(\mathbf A_u)|.
\end{equation*}
Physically, $R_{\mathrm{exp}}$ quantifies the intrinsic entropy generation rate of the open-loop system. It measures how fast the unstable state volume expands in the absence of control.

\subsection{Directed Information Rate}
To quantify the causal information flow from the unstable state sequence to the observation sequence in our feedback loop, we use directed information~\cite{Massey1990CAUSALITYFA,Charalambous2013}. The cumulative directed information up to time $t$ is defined by
\begin{equation} \label{eq:directed_info_def}
    I(\mathbf x^{t} \to \mathbf y^{t}) := \sum_{k=0}^t I(\mathbf x_k; \mathbf y_k | \mathbf y^{k-1}),
\end{equation}
where $I(\mathbf x_k; \mathbf y_k | \mathbf y^{k-1})$ denotes the conditional mutual information between the state and the observation at time $k$, given the past observation history.

The corresponding \textit{directed information rate} is
\begin{equation*}
    \mathcal{I}(\mathbf x \to \mathbf y) := \liminf_{t \to \infty} \frac{1}{t+1} I(\mathbf x^{t} \to \mathbf y^{t}),
\end{equation*}
which measures the average information conveyed over time by the observations about the state trajectory, in bits/step.
This rate serves as the fundamental measure of the sensing capability for state estimation and stabilization.

\begin{remark}
The general definition of directed information involves the mutual information between the state history $\mathbf x^{t}$ and the observation $\mathbf y_t$, i.e., $I(\mathbf x^{t}; \mathbf y_t | \mathbf y^{t-1})$.
Using the definition of conditional mutual information, this can be expanded as
\begin{equation*}
    I(\mathbf x^{t}; \mathbf y_t | \mathbf y^{t-1}) = h(\mathbf y_t | \mathbf y^{t-1}) - h(\mathbf y_t | \mathbf x^{t}, \mathbf y^{t-1}).
\end{equation*}
By the memoryless observation mechanism, the current observation is generated
according to the conditional law $p(\mathbf y_t|\mathbf x_t)$. Hence,
conditional on $\mathbf x_t$, the past states and past observations provide no
additional information about $\mathbf y_t$, which implies
\begin{equation*}
    h(\mathbf y_t|\mathbf x^{t},\mathbf y^{t-1})
    =
    h(\mathbf y_t|\mathbf x_t,\mathbf y^{t-1}).
\end{equation*}
Therefore,
\begin{align*}
    I(\mathbf x^{t};\mathbf y_t|\mathbf y^{t-1})
    &= h(\mathbf y_t|\mathbf y^{t-1})
       - h(\mathbf y_t|\mathbf x_t,\mathbf y^{t-1}) \\
    &= I(\mathbf x_t;\mathbf y_t|\mathbf y^{t-1}).
\end{align*}
This justifies the simplified form used in \eqref{eq:directed_info_def}, where conditioning on the state history reduces to conditioning on the instantaneous state.
\end{remark}

\subsection{Definitions of Stabilizability and Observability}
In this paper, we distinguish between \textit{boundedness} (used for necessary conditions) and \textit{asymptotic convergence} (used for sufficient conditions). The definitions are formalized as follows.

\begin{definition}[Mean-Square Stabilizability] \label{def:ms_stab}
    The system is said to be \textit{mean-square stabilizable} if there exists a control strategy such that the second moment of the state remains bounded
    \begin{equation*}
        \limsup_{t \to \infty} \mathbb{E}[\|\mathbf x_t\|^2] < \infty.
    \end{equation*}
\end{definition}

\begin{definition}[Asymptotic Mean-Square Stabilizability] \label{def:asymp_ms_stab}
    The system is said to be \textit{asymptotically mean-square stabilizable} if there exists a control strategy such that for any initial state $\mathbf x_0$, the following two properties are satisfied:
    \begin{enumerate}
        \item Stability: $\forall \epsilon > 0, \exists \delta(\epsilon) > 0$ such that if $\mathbb{E}[\|\mathbf x_0\|^2] \le \delta$, then $\mathbb{E}[\|\mathbf x_t\|^2] \le \epsilon$ for all $t \in \mathbb{N}_0$.
        \item Attractivity: $\lim_{t \to \infty} \mathbb{E}[\|\mathbf x_t\|^2] = 0$.
    \end{enumerate}
\end{definition}

\begin{definition}[Mean-Square Observability] \label{def:ms_obs}
    The system is said to be \textit{mean-square observable} if there exists a causal estimation strategy providing a predictive estimate $\hat{\mathbf x}_t$, measurable with respect to $\mathbf y^{t-1}$, such that the second moment of the estimation error $\mathbf e_t := \mathbf x_t - \hat{\mathbf x}_t $ remains bounded in the mean-square sense
    \begin{equation*}
        \limsup_{t \to \infty} \mathbb{E}[\|\mathbf e_t\|^2] < \infty.
    \end{equation*}
\end{definition}

\begin{definition}[Asymptotic Mean-Square Observability] \label{def:asymp_ms_obs}
    The system is said to be \textit{asymptotically mean-square observable} if there exists a causal estimation strategy such that for any initial error $\mathbf e_0$, the following two properties are satisfied:
    \begin{enumerate}
        \item Stability: $\forall \epsilon > 0, \exists \delta(\epsilon) > 0$ such that if $\mathbb{E}[\|\mathbf e_0\|^2] \le \delta$, then $\mathbb{E}[\|\mathbf e_t\|^2] \le \epsilon$ for all $t \in \mathbb{N}_0$.
        \item Attractivity: $\lim_{t \to \infty} \mathbb{E}[\|\mathbf e_t\|^2] = 0$.
    \end{enumerate}
\end{definition}

\begin{remark}
We retain the standard Lyapunov asymptotic notions, with both stability and attractivity, because attractivity alone allows large transient mean-square excursions from small initial errors. Accordingly, the sufficiency results first prove attractivity and then add a local overshoot condition for stability; the converse uses the weaker bounded notions in Definitions~\ref{def:ms_stab} and~\ref{def:ms_obs} for generality.
\end{remark}

\section{Necessary Conditions} \label{sec:necessary}

This section establishes necessary information rate conditions for
mean-square observability and stabilizability.
For noiseless dynamics, we show that the directed information rate from
the unstable state to the observations must satisfy
$\mathcal I(\mathbf z\to\mathbf y)\ge R_{\mathrm{exp}}$.
The same lower bound remains necessary under additive process
disturbances.
Because the directed information rate from the unstable state is generally
difficult to compute, we develop an impossibility criterion for nonlinear
observation models with additive noise by deriving an upper bound on the rate
from the full state through the conditional output covariance.
Any such upper bound below $R_{\mathrm{exp}}$ certifies infeasibility.
In the linear Gaussian setting, the bound is attained and its exact value
is determined by the steady-state Riccati equation.
\subsection{Necessity Results} \label{subsec:theorems}

\begin{theorem}[Necessity for Observability] \label{thm-Necessary-Observability}
    Consider the noiseless linear system defined in \eqref{eq:system_dynamics}. If the system is mean-square observable, then the directed information rate from the unstable states to the observations must satisfy
    \begin{equation*}
       \mathcal{I}(\mathbf z \to \mathbf y) \ge R_{\mathrm{exp}}.
    \end{equation*}
\end{theorem}

\begin{IEEEproof}[Proof]
The proof relies on analyzing the evolution of the conditional differential entropy of the unstable modes. Let $h_t := h(\mathbf z_t | \mathbf y^{t-1})$ denote the prior differential entropy of the unstable state at time $t$ given past observations.

For the unstable dynamics $\mathbf z_{t+1} = \mathbf A_u \mathbf z_t + \mathbf B_u \mathbf u_t$, the conditional entropy evolves as
\begin{align}
    h(\mathbf z_{t+1} | \mathbf y^{t})
    &= h(\mathbf A_u \mathbf z_t + \mathbf B_u \mathbf u_t | \mathbf y^{t}) \nonumber \\
    &\overset{(a)}{=} h(\mathbf A_u \mathbf z_t | \mathbf y^{t}) \nonumber \\
    &\overset{(b)}{=} h(\mathbf z_t | \mathbf y^{t}) + \log_2 |\det \mathbf A_u| \nonumber \\
    &= h(\mathbf z_t | \mathbf y^{t}) + R_{\mathrm{exp}}, \label{eq:entropy_evolution}
\end{align}
where (a) holds because $\mathbf u_t$ is measurable with respect to $\mathbf y^{t-1}$ and hence known under the conditioning on $\mathbf y^{t}$, so the conditionally deterministic shift $\mathbf B_u\mathbf u_t$ does not change differential entropy. Equality (b) follows from the change-of-variables formula for differential entropy under the invertible linear transformation $\mathbf A_u$.
Using the identity $I(\mathbf z; \mathbf y | \mathbf x) = h(\mathbf z|\mathbf x) - h(\mathbf z|\mathbf y,\mathbf x)$, we relate the posterior entropy to the prior entropy through
\begin{equation}
    h(\mathbf z_t | \mathbf y^{t}) = h(\mathbf z_t | \mathbf y^{t-1}) - I(\mathbf z_t; \mathbf y_t | \mathbf y^{t-1}).
    \label{eq:posterior_prior}
\end{equation}
Substituting \eqref{eq:posterior_prior} into \eqref{eq:entropy_evolution}, we obtain the recursion for the prior entropy
\begin{equation}
    h_{t+1} = h_t + R_{\mathrm{exp}} - I(\mathbf z_t; \mathbf y_t | \mathbf y^{t-1}).
    \label{eq:recursion}
\end{equation}
Summing the recursion \eqref{eq:recursion} from $k=0$ to $t$ gives
\begin{equation*}
    \sum_{k=0}^{t} (h_{k+1} - h_k) = (t+1)R_{\mathrm{exp}} - \sum_{k=0}^{t} I(\mathbf z_k; \mathbf y_k | \mathbf y^{k-1}).
\end{equation*}
Recognizing the telescoping sum on the left and the definition of cumulative directed information on the right, we divide by $t+1$ to obtain
\begin{equation*}
    \frac{h_{t+1} - h_0}{t+1} = R_{\mathrm{exp}} - \frac{1}{t+1} I(\mathbf z^{t} \to \mathbf y^{t}).
\end{equation*}
Rearranging the directed information rate balance gives
\begin{equation}
    \frac{1}{t+1} I(\mathbf z^{t} \to \mathbf y^{t}) = R_{\mathrm{exp}} + \frac{h_0}{t+1} - \frac{h_{t+1}}{t+1}.
    \label{eq:rate_balance_final}
\end{equation}

We now show that $\limsup_{t \to \infty} \frac{h_{t+1}}{t+1} \le 0$. First, we establish the boundedness of the estimation error for the unstable subsystem. Partitioning the transformed error vector according to the state decomposition gives
\begin{equation*}
    \mathbf T\mathbf e_t
    =
    \begin{bmatrix}
        \boldsymbol\epsilon_t\\
        \boldsymbol\eta_t
    \end{bmatrix},
\end{equation*}
where $\boldsymbol\epsilon_t:=\mathbf z_t-\hat{\mathbf z}_t$ is the
error in the unstable state and $\boldsymbol\eta_t$ is the error in the stable state.
 The assumption of \textit{Mean-Square Observability} implies that the total estimation error $\mathbf e_t := \mathbf x_t - \hat{\mathbf x}_t$ is uniformly bounded, i.e., $\limsup_{t \to \infty} \mathbb{E}[\|\mathbf e_t\|^2] \le C_{e} < \infty$.
Since the squared norm of the unstable component $\boldsymbol\epsilon_t$
is bounded by the total transformed error, we have
\begin{equation*}
    \mathbb{E}[\|\boldsymbol\epsilon_t\|^2]
    \le \mathbb{E}[\|\mathbf T\mathbf e_t\|^2]
    \le \|\mathbf T\|^2 \mathbb{E}[\|\mathbf e_t\|^2]
    \le \|\mathbf T\|^2 C_{e} := C.
\end{equation*}
Thus, the unstable error variance is uniformly bounded by a constant $C$.

We then bound the entropy $h_{t+1} = h(\mathbf z_{t+1} | \mathbf y^{t})$ using the property that the Gaussian distribution maximizes entropy for a given covariance and the AM-GM inequality relating determinant to trace as follows
\begin{align*}
    h_{t+1}
    &= h(\mathbf z_{t+1} | \mathbf y^{t}) \notag \\
    &\stackrel{(a)}{=} h(\mathbf z_{t+1} - \hat{\mathbf z}_{t+1} | \mathbf y^{t}) \notag \\
    &= h(\boldsymbol{\epsilon}_{t+1} | \mathbf y^{t}) \notag \\
    &\stackrel{(b)}{\le} h(\boldsymbol{\epsilon}_{t+1}) \notag \\
    &\stackrel{(c)}{\le} \frac{1}{2} \log_2 \left( (2\pi e)^{n_u} \det(\boldsymbol\Sigma_{e}) \right) \notag \\
    &\stackrel{(d)}{\le} \frac{n_u}{2} \log_2 \left( \frac{2\pi e}{n_u} \mathrm{Tr}(\boldsymbol\Sigma_{e}) \right) \notag \\
    &\stackrel{(e)}{\le} \frac{n_u}{2} \log_2 \left( \frac{2\pi e}{n_u} C \right) := C_h,
\end{align*}
\noindent where
\begin{itemize}
    \item[(a)] Translation invariance of conditional differential entropy, since $\hat{\mathbf z}_{t+1}$ is determined by $\mathbf y^{t}$.
    \item[(b)] Conditioning reduces entropy.
    \item[(c)] The Gaussian distribution maximizes differential entropy for a fixed covariance matrix $\boldsymbol\Sigma_{e}$.
    \item[(d)] By the AM-GM inequality\cite{beckenbach1961inequalities}, $\det(\boldsymbol\Sigma) \le (\frac{1}{n}\mathrm{Tr}(\boldsymbol\Sigma))^n$.
    \item[(e)] The trace of the error covariance is bounded by $C$ due to the observability assumption.
\end{itemize}
Thus, $C_h$ is a constant derived from the uniform error bound $C$. $h_0 :=h(\mathbf z_{0} | \mathbf y^{-1})= h(\mathbf z_0) $ is finite because the initial state $\mathbf x_0$ is assumed to have a valid probability density function with a finite second moment. Hence, $h_0$ is bounded above by the entropy of a Gaussian distribution with the same covariance. Crucially, both $C_h$ and the initial entropy $h_0$ are finite constants independent of $t$.

Taking the limit inferior ($t \to \infty$) of \eqref{eq:rate_balance_final} gives
\begin{align*}
    \liminf_{t \to \infty} \frac{1}{t+1} I(\mathbf z^{t} \to \mathbf y^{t})
    &= R_{\mathrm{exp}} + 0 - \limsup_{t \to \infty} \frac{h_{t+1}}{t+1} \nonumber \\
    &\ge R_{\mathrm{exp}} - \lim_{t \to \infty} \frac{C_h}{t+1} \nonumber \\
    &= R_{\mathrm{exp}}.\notag
\end{align*}
This completes the proof.
\end{IEEEproof}

\begin{theorem}[Necessity for Stabilizability] \label{thm-Necessary-Stabilizability}
    If the system is mean-square stabilizable, then the directed information rate from the unstable states to the observations must satisfy
    \begin{equation*}
        \mathcal{I}(\mathbf z \to \mathbf y) \ge R_{\mathrm{exp}}.
    \end{equation*}
\end{theorem}

\begin{IEEEproof}
The proof follows the same entropy balance argument as Theorem~\ref{thm-Necessary-Observability}, with state boundedness replacing boundedness of the estimation error. Details are given in \apprefappendix{app:proof_thm_stab}{A}.
\end{IEEEproof}

Theorems~\ref{thm-Necessary-Observability} and~\ref{thm-Necessary-Stabilizability} rely only on entropy identities and boundedness assumptions; they do not require the observation law to be linear or to have additive noise.
Thus, for any observation mechanism described by a conditional density $p(\mathbf y_t|\mathbf x_t)$, mean-square observability or stabilizability is impossible unless
\[
    \mathcal{I}(\mathbf z \to \mathbf y) \ge R_{\mathrm{exp}}.
\]

\subsection{Extension to Noisy Linear Dynamics} \label{subsec:noisy}

The analysis so far has focused on the noiseless dynamics.
The following proposition establishes that the rate bound derived in Theorem~\ref{thm-Necessary-Observability} remains necessary for mean-square observability and stabilizability in the presence of additive process disturbances.

\begin{proposition}[Robustness to Process Noise] \label{prop:noisy_necessity}
    Consider the model with additive process disturbances given by
    \begin{equation} \label{eq:system_dynamics_noisy}
        \mathbf x_{t+1} = \mathbf A\mathbf x_t + \mathbf B\mathbf u_t + \mathbf w_t, \quad \forall t \in \mathbb{N}_0,
    \end{equation}
    where $\{\mathbf w_t\}_{t\in \mathbb{N}_0}$ is a sequence of mutually independent continuous random vectors with finite differential entropy. For each $t$, $\mathbf w_t$ is independent of $\mathbf x_t$, $\mathbf u_t$, and the observation history $\mathbf y^{t}$.
    If the system is mean-square observable (or stabilizable), the necessary condition derived for the noiseless case remains valid as the lower bound
    \begin{equation*}
       \mathcal{I}(\mathbf z \to \mathbf y) \ge R_{\mathrm{exp}}.\notag
    \end{equation*}
\end{proposition}

\begin{IEEEproof}
    The proof mirrors that of Theorem~\ref{thm-Necessary-Observability}, with a key modification in the entropy evolution step to account for the noise.
    The conditional entropy of the next unstable state satisfies
    \begin{align}
        h(\mathbf z_{t+1} | \mathbf y^{t})
        &= h(\mathbf A_u \mathbf z_t + \mathbf B_u \mathbf u_t + \mathbf w_{z,t} | \mathbf y^{t}) \notag \\
        &\stackrel{(a)}{=} h(\mathbf A_u \mathbf z_t + \mathbf w_{z,t} | \mathbf y^{t}) \notag \\
        &\stackrel{(b)}{\ge} h(\mathbf A_u \mathbf z_t | \mathbf y^{t}) \notag \\
        &= h(\mathbf z_t | \mathbf y^{t}) + \log_2 |\det \mathbf A_u| \notag \\
        &= h(\mathbf z_t | \mathbf y^{t}) + R_{\mathrm{exp}}, \label{eq:noisy_entropy_inequality}
    \end{align}
    where:
    \begin{itemize}
        \item[(a)] Translation invariance.
        \item[(b)] This step relies on the entropy inequality $h(\mathbf a+\mathbf b) \ge h(\mathbf a)$ for independent random vectors $\mathbf a$ and $\mathbf b$. Here, $\mathbf w_{z,t}$ is independent of $(\mathbf z_t,\mathbf y^{t})$, and therefore remains independent of $\mathbf A_u\mathbf z_t$ under the conditioning on $\mathbf y^{t}$, where $\mathbf T\mathbf w_t=[\mathbf w_{z,t}^T,\mathbf w_{s,t}^T]^T$.
    \end{itemize}
    The equality in the noiseless case \eqref{eq:entropy_evolution} now becomes an inequality. Substituting the mutual information relation $h(\mathbf z_t | \mathbf y^{t}) = h(\mathbf z_t | \mathbf y^{t-1}) - I(\mathbf z_t; \mathbf y_t | \mathbf y^{t-1})$ into \eqref{eq:noisy_entropy_inequality}, we obtain the recurrence
    \begin{equation*}
        h_{t+1} \ge h_t + R_{\mathrm{exp}} - I(\mathbf z_t; \mathbf y_t | \mathbf y^{t-1}).\notag
    \end{equation*}
    Summing from $k=0$ to $t$ and dividing by $t+1$ gives
    \begin{equation*}
        \frac{1}{t+1} I(\mathbf z^{t} \to \mathbf y^{t}) \ge R_{\mathrm{exp}} + \frac{h_0}{t+1} - \frac{h_{t+1}}{t+1}.\notag
    \end{equation*}
    Given the mean-square observability (or stabilizability) assumption, the state/error covariance is bounded, ensuring that $\limsup_{t \to \infty} \frac{h_{t+1}}{t+1} \le 0$ and $\lim_{t \to \infty} \frac{h_{0}}{t+1} = 0$ (as derived in the proofs of Theorems~\ref{thm-Necessary-Observability} and~\ref{thm-Necessary-Stabilizability}).
    Taking the limit inferior gives
    \begin{equation*}
        \liminf_{t \to \infty} \frac{1}{t+1} I(\mathbf z^{t} \to \mathbf y^{t}) \ge R_{\mathrm{exp}}.
    \end{equation*}
\end{IEEEproof}

\subsection{A Computable Impossibility Criterion}
\label{subsec:computable_relaxations}

The directed information rate from the unstable state is generally unavailable in
closed form. A direct infeasibility test can instead be obtained from the
computable upper bound $\mathcal I_{\max}$ on the directed
information rate from the full state developed below.
Consider the noisy dynamical system in
\eqref{eq:system_dynamics_noisy}, observed through
\begin{equation}
\label{eq:nonlinear_additive_observation}
    \mathbf y_t=g_t(\mathbf x_t)+\mathbf v_t,
    \qquad t\in\mathbb N_0,
\end{equation}
where $g_t:\mathbb R^n\to\mathbb R^p$ is deterministic and
possibly nonlinear. The sequence
$\{\mathbf v_t\}_{t\in\mathbb N_0}$ is mutually independent, has finite
differential entropies, and is independent of the process noise; moreover,
$\mathbf v_t$ is independent of
$(\mathbf x_t,\mathbf y^{t-1})$ for every $t$.
Define the conditional output covariance
\begin{equation*}
    \boldsymbol\Gamma_t
    :=
    \mathrm{Cov}(\mathbf y_t\mid\mathbf y^{t-1}).
\end{equation*}
Let $\{\mathbf S_t\succ\mathbf 0\}_{t\in\mathbb N_0}$ be a
deterministic sequence satisfying
\begin{equation}
\label{eq:conditional_output_covariance_bound}
    \boldsymbol\Gamma_t
    \preceq \mathbf S_t,
    \qquad \text{almost surely}.
\end{equation}
Define
\begin{equation*}
\begin{aligned}
    \mathcal I_{\max}
    &:={}
    \liminf_{T\to\infty}\frac1{T+1}
    \sum_{t=0}^{T}
    \left[
        \frac12\log_2\det(2\pi e\mathbf S_t)
        -h(\mathbf v_t)
    \right].
\end{aligned}
\end{equation*}

\begin{proposition}[Upper Bound for Nonlinear Observations with Additive Noise]
\label{prop:upper_bound_additive}
Under the preceding observation model \eqref{eq:nonlinear_additive_observation} and covariance bound \eqref{eq:conditional_output_covariance_bound}, the directed information rate from the full state is bounded above by $\mathcal I_{\max}$
\begin{equation*}
    \mathcal I(\mathbf x\to\mathbf y)
    \le \mathcal I_{\max}.
\end{equation*}
Moreover, if $\mathbf v_t\sim\mathcal N(\mathbf 0,\mathbf R)$ with
$\mathbf R\succ\mathbf 0$, then
\begin{align*}
    \mathcal I_{\max}
    =\liminf_{T\to\infty}\frac1{T+1}
    \sum_{t=0}^{T}
    \frac12\log_2\det\!\left(
        \mathbf S_t\mathbf R^{-1}
    \right).
\end{align*}
\end{proposition}

\begin{IEEEproof}
Under the observation model in
\eqref{eq:nonlinear_additive_observation},
\begin{equation*}
I(\mathbf x_t;\mathbf y_t|\mathbf y^{t-1})
=h(\mathbf y_t|\mathbf y^{t-1})-h(\mathbf v_t).
\end{equation*}
Applying the Gaussian maximum entropy property conditionally on each
fixed observation history and then averaging over that history gives
\begin{align*}
h(\mathbf y_t|\mathbf y^{t-1})
&\le
\frac12\mathbb E_{\mathbf y^{t-1}}\!\left[
\log_2\det\!\left(
2\pi e\,
\boldsymbol\Gamma_t
\right)
\right]\\
&\le
\frac12\log_2\det(2\pi e\mathbf S_t).
\end{align*}
The last inequality follows from
$\boldsymbol\Gamma_t
\preceq\mathbf S_t$ almost surely and the monotonicity of
$\log\det(\cdot)$ on the positive-definite cone. Since
$\mathbf S_t$ is deterministic, the resulting upper bound is
independent of the observation history.
Therefore, for every $T$,
\begin{align*}
&\frac{1}{T+1}
I(\mathbf x^{T}\to\mathbf y^{T})=
\frac{1}{T+1}
\sum_{t=0}^{T}
I(\mathbf x_t;\mathbf y_t|\mathbf y^{t-1})\\
&\le
\frac{1}{T+1}
\sum_{t=0}^{T}
\left[
\frac12\log_2\det(2\pi e\mathbf S_t)
-h(\mathbf v_t)
\right].
\end{align*}
Consequently,
\begin{align*}
\mathcal I(\mathbf x\to\mathbf y)
&=
\liminf_{T\to\infty}
\frac{1}{T+1}
I(\mathbf x^{T}\to\mathbf y^{T})\\
&\le
\liminf_{T\to\infty}
\frac{1}{T+1}
\sum_{t=0}^{T}
\left[
\frac12\log_2\det(2\pi e\mathbf S_t)
-h(\mathbf v_t)
\right]
\\
&=\mathcal I_{\max}.
\end{align*}
The Gaussian expression follows from
$h(\mathbf v_t)=\frac12\log_2\det(2\pi e\mathbf R)$.
\end{IEEEproof}

Whenever mean-square observability or stabilizability is feasible, the
necessary results and Proposition~\ref{prop:upper_bound_additive} imply
\begin{equation}
\label{eq:computable_impossibility_chain}
R_{\mathrm{exp}}
\le
\mathcal I(\mathbf z\to\mathbf y)
\le
\mathcal I(\mathbf x\to\mathbf y)
\le \mathcal I_{\max}.
\end{equation}
Consequently, $\mathcal I_{\max}<R_{\mathrm{exp}}$ certifies that
mean-square observability and stabilizability are infeasible. The
following linear Gaussian specialization evaluates this upper bound
exactly.

Consider the system with independent additive Gaussian process and measurement noise given by
\begin{equation} \label{eq:lg_system_noisy}
    \begin{cases}
        \mathbf x_{t+1} = \mathbf A\mathbf x_t + \mathbf B\mathbf u_t + \mathbf w_t, & \mathbf w_t \sim \mathcal{N}(\mathbf 0, \mathbf W) \\
        \mathbf y_t = \mathbf C\mathbf x_t + \mathbf v_t, & \mathbf v_t \sim \mathcal{N}(\mathbf 0, \mathbf R)
    \end{cases},
\end{equation}
where the vectors $\mathbf w_t \in \mathbb{R}^n$ and $\mathbf v_t \in \mathbb{R}^p$ represent the independent Gaussian process and measurement noise sequences, respectively.
The observation mechanism is linear and characterized by the matrix $\mathbf C \in \mathbb{R}^{p \times n}$.
The noise covariance matrices are assumed to be positive definite, i.e., $\mathbf W \succ \mathbf 0$\footnote{While general formulations typically allow $\mathbf W$ to be positive semidefinite ($\mathbf W \succeq \mathbf 0$) provided that the pair $(\mathbf A, \mathbf W^{1/2})$ is stabilizable, we assume strict positive definiteness here for simplicity of exposition.} and $\mathbf R \succ \mathbf 0$.
The initial state $\mathbf x_0$ is Gaussian with mean $\boldsymbol{\mu}_0$ and covariance $\boldsymbol\Sigma_0$, independent of the noise sequences; equivalently, under the initial estimate $\hat{\mathbf x}_{0|-1}=\boldsymbol{\mu}_0$, the initial prediction error covariance is $\boldsymbol\Sigma_0$.

We use the standard conditions for Kalman covariance convergence. The pair
$(\mathbf A,\mathbf C)$ is \emph{detectable} if there exists
$\mathbf L$ such that
$\mathbf A+\mathbf L\mathbf C$ is Schur, and
$(\mathbf A,\mathbf W^{1/2})$ is \emph{stabilizable} if there exists
$\mathbf K$ such that
$\mathbf A+\mathbf W^{1/2}\mathbf K$ is Schur
\cite[Secs.~4.4 and 5.2]{anderson2005optimal}.

\begin{corollary}[Linear Gaussian Exact Upper Bound]
\label{cor:lg_capacity}
    Consider the linear Gaussian system defined in \eqref{eq:lg_system_noisy}. Assume that the pair $(\mathbf A, \mathbf C)$ is detectable and the pair $(\mathbf A, \mathbf W^{1/2})$ is stabilizable.
    Then the Kalman prediction error covariance converges from any initial condition $\boldsymbol\Sigma_0 \succeq 0$ to a unique positive-definite steady-state matrix $\boldsymbol\Sigma$ that satisfies the \textit{Discrete Algebraic Riccati Equation (DARE)}
    \begin{equation*}
        \boldsymbol\Sigma = \mathbf A\boldsymbol\Sigma\mathbf A^T + \mathbf W - \mathbf A\boldsymbol\Sigma\mathbf C^T \left( \mathbf C\boldsymbol\Sigma\mathbf C^T + \mathbf R \right)^{-1} \mathbf C\boldsymbol\Sigma\mathbf A^T.
    \end{equation*}
    The upper bound in Proposition~\ref{prop:upper_bound_additive} is
    attained, and the directed information rate from the full state is
    \begin{equation} \label{eq:sensing_capacity_formula}
    \begin{aligned}
        \mathcal{I}(\mathbf x \to \mathbf y)
        = \frac{1}{2} \log_2 \det \left( \mathbf I_p+ \mathbf C\boldsymbol\Sigma\mathbf C^T \mathbf R^{-1} \right).
    \end{aligned}
    \end{equation}
\end{corollary}

\begin{IEEEproof}
Define
\begin{equation*}
    \hat{\mathbf x}_{t|t-1}
    :=\mathbb E[\mathbf x_t|\mathbf y^{t-1}],
    \qquad
    \boldsymbol\Sigma_t
    :=\mathrm{Cov}(\mathbf x_t|\mathbf y^{t-1}).
\end{equation*}
In the linear Gaussian model, $\boldsymbol\Sigma_t$ is
deterministic and follows the Kalman covariance recursion.
The conditional output distribution is Gaussian with covariance
\begin{equation*}
    \mathbf S_t
    =\mathrm{Cov}(\mathbf y_t|\mathbf y^{t-1})
    =\mathbf C\boldsymbol\Sigma_t
      \mathbf C^T+\mathbf R.
\end{equation*}
Thus, Proposition~\ref{prop:upper_bound_additive} applies with this
deterministic sequence $\{\mathbf S_t\}$.
Since $\mathbf y_t|\mathbf y^{t-1}$ and $\mathbf v_t$ are Gaussian,
\begin{equation*}
\begin{aligned}
I(\mathbf x_t;\mathbf y_t|\mathbf y^{t-1})
&=h(\mathbf y_t|\mathbf y^{t-1})-h(\mathbf v_t)\\
&=\frac12\log_2\det(\mathbf S_t\mathbf R^{-1}).
\end{aligned}
\end{equation*}

If $(\mathbf A,\mathbf C)$ is detectable and
$(\mathbf A,\mathbf W^{1/2})$ is stabilizable, the Kalman prediction covariance
converges to the unique positive-definite DARE solution
$\boldsymbol\Sigma$ \cite[Ch.~5]{anderson2005optimal}. Hence
$\mathbf S_t\to
\mathbf S:=\mathbf C\boldsymbol\Sigma
\mathbf C^T+\mathbf R$. Ces\`aro convergence then gives
\begin{align*}
\mathcal I(\mathbf x\to\mathbf y)
&=\frac12\log_2\det(\mathbf S\mathbf R^{-1})\\
&=\frac12\log_2\det\!\left(
\mathbf I_p+
\mathbf C\boldsymbol\Sigma
\mathbf C^T\mathbf R^{-1}
\right),
\end{align*}
which proves \eqref{eq:sensing_capacity_formula}.
\end{IEEEproof}

\section{Sufficient Conditions} \label{sec:sufficient}
All sufficiency results in this section are established for the
noiseless dynamics in \eqref{eq:system_dynamics}. They give conditions under
which the available sensing information drives the estimation error and
closed-loop state to zero in mean square.
We introduce the posterior non-Gaussianity rate $R_{\mathrm{NG}}$ to
quantify the asymptotic gap between the posterior entropy and its covariance-matched Gaussian.
Under uniform bounds on the posterior covariance and its fluctuations, the
condition
$\mathcal I(\mathbf z\to\mathbf y)>
R_{\mathrm{exp}}+R_{\mathrm{NG}}$
guarantees mean-square attractivity of the estimation error.
For a stabilizable plant, certainty-equivalence feedback further guarantees
mean-square attractivity of the closed-loop state.
Together with a local mean-square overshoot bound, these results yield
asymptotic mean-square observability and stabilizability.

When each observation is generated by a nonlinear function of the
unstable state plus additive noise, we further derive a lower bound on the
directed information rate based on entropy power. This bound provides a direct
way to verify the rate condition in the sufficiency theorems
without evaluating directed information exactly.

We then study the class with $R_{\mathrm{NG}}=0$, for which the
sufficient information threshold reduces to $R_{\mathrm{exp}}$.
Linear Gaussian observations and their invertible readout variants
provide direct examples.
We also derive curvature conditions on the likelihood and prior that establish
$R_{\mathrm{NG}}=0$ through eventual posterior log-concavity.

\subsection{Sufficiency Results} \label{subsec:sufficiency_results}
Following the KL measure of non-Gaussianity in
\cite{tulino2006nonGaussianity,reeves2019replica}, define the random
posterior density
\begin{equation*}
p_t(\boldsymbol\xi)
:=
p_{\mathbf z_t\mid\mathbf y^{t}}
(\boldsymbol\xi\mid\mathbf y^{t}),
\end{equation*}
and let $q_t$ denote the density of
$\mathcal N(\hat{\mathbf z}_t,\mathbf P_t)$, the Gaussian distribution
with the same conditional mean and covariance. Both densities are random
through $\mathbf y^t$. Define the resulting random divergence by
\begin{equation*}
    D_t:=D_{\mathrm{KL}}(p_t\|q_t).
\end{equation*}

\begin{definition}[Posterior Non-Gaussianity Rate]
\label{def:posterior_non_gaussianity_rate}
The posterior non-Gaussianity rate is defined as
\begin{equation*}
    R_{\mathrm{NG}}
    :=
    \limsup_{t\to\infty}
    \frac{1}{t+1}\mathbb E[D_t].
\end{equation*}
The expectation is over the observation history $\mathbf y^t$.
The rate takes values in the
extended interval $[0,\infty]$ and is measured in bits/step. In
particular, $R_{\mathrm{NG}}=0$ means that the expected
divergence between the posterior and its moment-matched Gaussian grows sublinearly with time.
\end{definition}

\begin{assumption} \label{assump:cond_number}
The posterior covariance family is assumed to be uniformly regular in the following sense: there exist constants $\kappa$, $M$, and $V$ such that, for all $t$, almost surely,
\begin{align*}
    &\mathbf P_t\succ\mathbf 0,\qquad
    \frac{\lambda_{\max}(\mathbf P_t)}
    {\lambda_{\min}(\mathbf P_t)}
    \le \kappa,\\
    &\mathrm{Tr}(\mathbf P_t)\le M,\qquad
    \mathrm{Var}\!\left[
        \log_2\det(\mathbf P_t)
    \right]\le V.
\end{align*}
\end{assumption}

Assumption~\ref{assump:cond_number} excludes
covariance pathologies that would otherwise prevent contraction of the posterior covariance volume
from implying convergence in mean square.
The trace bound and the bound on the condition number prevent unbounded
uncertainty and highly anisotropic covariance, while the variance bound
rules out expected log-volume collapse caused only by rare observation
histories. These conditions require neither a Gaussian posterior nor a
specific covariance recursion.

\begin{lemma}[Log-volume collapse implies mean-square collapse]\label{lemma:log_volume_to_mse}
    Let $\{\mathbf P_t\}_{t\ge0}$ be a sequence of random positive-definite matrices satisfying Assumption~\ref{assump:cond_number} with dimension $n_u$.
    If
    \begin{equation*}
        \mathbb{E}\!\left[\log_2\det(\mathbf P_t)\right]\to-\infty,
    \end{equation*}
    then
    \begin{equation*}
        \mathbb{E}\!\left[\mathrm{Tr}(\mathbf P_t)\right]\to0.
    \end{equation*}
\end{lemma}

\begin{IEEEproof}
    Let $L_t:=\log_2\det(\mathbf P_t)$.
    Since $\mathbb{E}[L_t]\to-\infty$ and $\mathrm{Var}(L_t)\le V$, Chebyshev's inequality implies that, for every fixed $a\in\mathbb{R}$,
    \begin{equation*}
        \mathbb{P}(L_t>a)
        \le
        \frac{V}{(a-\mathbb{E}[L_t])^2}
        \to0.
    \end{equation*}
    Hence $\det(\mathbf P_t)\to0$ in probability.
    The bound on the condition number implies that if $\mathrm{Tr}(\mathbf P_t)>\epsilon$, then
    \begin{equation*}
        \epsilon < \sum_{i=1}^{n_u} \lambda_i \le n_u \lambda_{\max}(\mathbf P_t) \le n_u \kappa \lambda_{\min}(\mathbf P_t),
    \end{equation*}
    where $\lambda_i$ are the eigenvalues of $\mathbf P_t$.
    Therefore $\mathbb{P}(\mathrm{Tr}(\mathbf P_t)>\epsilon)\le \mathbb{P}(\det(\mathbf P_t)\ge(\frac{\epsilon}{n_u\kappa})^{n_u})\to0$.
    Since $\mathrm{Tr}(\mathbf P_t)\le M$ almost surely, the sequence $\{\mathrm{Tr}(\mathbf P_t)\}_{t\ge0}$ is uniformly integrable. Its convergence to zero in probability therefore implies
    $\mathbb{E}[\mathrm{Tr}(\mathbf P_t)]\to0$.
\end{IEEEproof}

\begin{theorem}
[Sufficiency for Mean-Square Attractivity of the Estimation Error]
\label{thm-Sufficiency-Observability}
Under Assumption~\ref{assump:cond_number}, if the directed information rate satisfies
\begin{equation*}
    \mathcal{I}(\mathbf z\to\mathbf y)
    >
    R_{\mathrm{exp}}+R_{\mathrm{NG}},
\end{equation*}
then there exists a causal estimator for the system defined in
\eqref{eq:system_dynamics} that achieves mean-square attractivity of the
estimation error, i.e.,
\begin{equation*}
    \lim_{t\to\infty}
    \mathbb{E}\!\left[
        \|\mathbf e_t\|^2
    \right]
    =0.
\end{equation*}
\end{theorem}

\begin{IEEEproof}
We first analyze the evolution of the conditional differential entropy
of the unstable state. The conditional entropies below are averaged over
their random conditioning histories. For every $k\ge0$, the control input $\mathbf u_k$
is known under conditioning on $\mathbf y^{k}$. Hence,
\begin{align}
h(\mathbf z_{k+1}|\mathbf y^{k})
&=
h(\mathbf A_u\mathbf z_k
+\mathbf B_u\mathbf u_k|\mathbf y^{k})
\notag\\
&=
h(\mathbf z_k|\mathbf y^{k})
+\log_2|\det(\mathbf A_u)|
\notag\\
&=
h(\mathbf z_k|\mathbf y^{k-1}) -I(\mathbf z_k;\mathbf y_k|\mathbf y^{k-1})
+R_{\mathrm{exp}}.
\label{eq:entropy_recursion}
\end{align}
The second equality follows from translation invariance and the
change-of-variables formula for differential entropy, while the last
equality follows from the conditional mutual information identity.

Summing \eqref{eq:entropy_recursion} from $k=0$ to $k=t-1$, with
$\mathbf y^{-1}$ denoting the empty observation history, gives
\begin{align}
h(\mathbf z_t|\mathbf y^{t-1})
&=
h(\mathbf z_0)
+tR_{\mathrm{exp}}
-\sum_{k=0}^{t-1}
I(\mathbf z_k;\mathbf y_k|\mathbf y^{k-1})
\notag\\
&=
h(\mathbf z_0)
+tR_{\mathrm{exp}}
-I(\mathbf z^{t-1}\to\mathbf y^{t-1}).
\label{eq:terminal_entropy}
\end{align}
The strict rate condition can hold only if $R_{\mathrm{NG}}<\infty$;
otherwise, its right-hand side equals $+\infty$. Define the strict rate
margin
\begin{equation*}
\delta
:=
\mathcal{I}(\mathbf z\to\mathbf y)
-R_{\mathrm{exp}}-R_{\mathrm{NG}}
>0,
\end{equation*}
and fix a constant $\epsilon$ such that $0<\epsilon<\delta/3$. The same
deterministic $\epsilon$ is used in the two asymptotic bounds below.

By an index shift in the definition of the directed information rate,
\begin{equation*}
\mathcal I(\mathbf z\to\mathbf y)
=
\liminf_{t\to\infty}
\frac{1}{t}
I(\mathbf z^{t-1}\to\mathbf y^{t-1}).
\end{equation*}
Hence, for the chosen $\epsilon$, there exists a time
$t_1=t_1(\epsilon)<\infty$ such that, for every $t\ge t_1$,
\begin{equation*}
I(\mathbf z^{t-1}\to\mathbf y^{t-1})
\ge
t\bigl(\mathcal I(\mathbf z\to\mathbf y)-\epsilon\bigr).
\end{equation*}
Substituting this bound into \eqref{eq:terminal_entropy} and then
conditioning additionally on $\mathbf y_t$ give
\begin{align}
h(\mathbf z_t|\mathbf y^{t})
&\le
h(\mathbf z_t|\mathbf y^{t-1})
\notag\\
&\le
h(\mathbf z_0)
-t\bigl(
\mathcal I(\mathbf z\to\mathbf y)
-R_{\mathrm{exp}}-\epsilon
\bigr),
\quad t\ge t_1.
\label{eq:entropy_divergence}
\end{align}

By the definition of $R_{\mathrm{NG}}$ as a limit superior, for the same
$\epsilon$ there exists a time
$t_2=t_2(\epsilon)<\infty$ such that, for every $t\ge t_2$,
\begin{equation}
\label{eq:non_gaussianity_tail_bound}
\mathbb E[D_t]
\le
(t+1)(R_{\mathrm{NG}}+\epsilon)
<
\infty.
\end{equation}
Set
\begin{equation*}
t_0:=\max\{t_1(\epsilon),t_2(\epsilon)\}.
\end{equation*}
The times $t_1$, $t_2$, and $t_0$ are independent of
the realized observation history.

The condition requiring finite variance in
Assumption~\ref{assump:cond_number} implies
\begin{equation*}
\mathbb E\!\left[
\left|
\log_2\det
\mathbf P_t
\right|
\right]
<\infty.
\end{equation*}
For each fixed observation history $\mathbf y^{t}=\mathbf a^{t}$, the
moment-matched Gaussian $q_t$ satisfies
\begin{align*}
D_t
&=
h(q_t)-h(p_t).
\end{align*}
Because $q_t$ has the same covariance as the posterior for the fixed
observation history, its entropy is
\begin{align*}
h(q_t)
&=
\frac12\log_2\!\Bigl(
(2\pi e)^{n_u}
\det\mathrm{Cov}(\mathbf z_t|
\mathbf y^{t}=\mathbf a^{t})
\Bigr).
\end{align*}
For every $t\ge t_0$, the preceding KL divergence and log-determinant
integrability bounds allow this identity to be integrated with respect to
$\mathbf y^{t}$. Taking expectation and rearranging the terms gives
\begin{align}
&\frac{1}{2}
\mathbb{E}\!\left[
\log_2\det(\mathbf P_t)
\right]
+\frac{n_u}{2}\log_2(2\pi e)
=
h(\mathbf z_t|\mathbf y^{t})
+
\mathbb{E}[D_t],
\label{eq:expected_entropy_deficit}
\end{align}
where
\begin{equation*}
h(\mathbf z_t|\mathbf y^{t})
:=
\int h(p_t)
p_{\mathbf y^{t}}(\mathbf y^{t})
\,\mathrm d\mathbf y^{t}.
\end{equation*}
Combining \eqref{eq:entropy_divergence},
\eqref{eq:non_gaussianity_tail_bound}, and
\eqref{eq:expected_entropy_deficit}, for every $t\ge t_0$ we obtain
\begin{align*}
\frac{1}{2}
\mathbb{E}\!\left[
\log_2\det(\mathbf P_t)
\right]
&\le
h(\mathbf z_0)
-t\bigl(
\mathcal I(\mathbf z\to\mathbf y)
-R_{\mathrm{exp}}-\epsilon
\bigr)\\
&+(t+1)(R_{\mathrm{NG}}+\epsilon)
-\frac{n_u}{2}\log_2(2\pi e)\\
&=
C_\epsilon-t(\delta-2\epsilon) < C_\epsilon-t\epsilon
\to-\infty,
\end{align*}
where
\begin{equation*}
C_\epsilon
:=
h(\mathbf z_0)+R_{\mathrm{NG}}+\epsilon
-\frac{n_u}{2}\log_2(2\pi e)
<\infty.
\end{equation*}
Therefore,
\begin{equation*}
\mathbb{E}\!\left[
\log_2\det(\mathbf P_t)
\right]
\to-\infty.
\end{equation*}

Applying Lemma~\ref{lemma:log_volume_to_mse} under
Assumption~\ref{assump:cond_number} gives
\begin{equation*}
\mathbb{E}\!\left[
\mathrm{Tr}(\mathbf P_t)
\right]
\to0.
\end{equation*}
For the posterior mean estimator,
$\boldsymbol{\epsilon}_t:=\mathbf z_t-\hat{\mathbf z}_t$, and the tower
property gives
\begin{equation*}
\mathbb{E}\!\left[\|\boldsymbol{\epsilon}_t\|^2\right]
=
\mathbb{E}\!\left[
\mathrm{Tr}(\mathbf P_t)
\right]
\to0.
\end{equation*}

For the stable modes, choose the estimator to propagate the stable prediction with the same known control input. Since $\mathbf A_s$ is Schur and the dynamics are noiseless, the stable estimation error satisfies $\boldsymbol{\eta}_{t+1}=\mathbf A_s\boldsymbol{\eta}_t$, and hence $\mathbb{E}[\|\boldsymbol{\eta}_t\|^2]\to0$. Therefore,
\begin{equation*}
    \mathbb{E}[\|\mathbf e_t\|^2]
    \le
    \|\mathbf T^{-1}\|^2
    \left(
        \mathbb{E}[\|\boldsymbol\epsilon_t\|^2]
        +
        \mathbb{E}[\|\boldsymbol\eta_t\|^2]
    \right)
    \to0.
\end{equation*}
Thus, the estimation error in the original coordinates converges to zero in mean square,
which proves the attractivity condition.

\end{IEEEproof}

\begin{remark}
If $R_{\mathrm{NG}}=0$, the rate condition in
Theorem~\ref{thm-Sufficiency-Observability} reduces to
\begin{equation*}
\mathcal{I}(\mathbf z\to\mathbf y)>R_{\mathrm{exp}}.
\end{equation*}
Thus, under Assumption~\ref{assump:cond_number}, the expansion rate is the
only asymptotic information rate threshold whenever the expected
divergence between the posterior and its moment-matched Gaussian grows sublinearly with time. This case does
not require the posterior itself to be Gaussian.
\end{remark}

\begin{theorem}[Sufficiency for Mean-Square Attractivity of the Closed-Loop State]\label{thm-Sufficiency-Stabilizability}
    Under Assumption~\ref{assump:cond_number}, suppose that the pair
$(\mathbf A, \mathbf B)$ is stabilizable\footnote{Mathematically, the pair $(\mathbf A, \mathbf B)$ is stabilizable if there exists a feedback gain matrix $\mathbf K$ such that the closed-loop matrix $\mathbf A + \mathbf B\mathbf K$ is Schur stable, i.e., all its eigenvalues satisfy $|\lambda_i(\mathbf A + \mathbf B\mathbf K)| < 1$.}. If
    \begin{equation*}
        \mathcal{I}(\mathbf z \to \mathbf y)
        > R_{\mathrm{exp}}+R_{\mathrm{NG}},
    \end{equation*}
    then there exists a certainty-equivalence control law that achieves
mean-square attractivity of the closed-loop state, i.e.,
    \begin{equation*}
        \lim_{t\to\infty}\mathbb{E}[\|\mathbf x_t\|^2]=0.
    \end{equation*}
\end{theorem}

\begin{IEEEproof}
    Theorem~\ref{thm-Sufficiency-Observability} applies under the stated
    rate and posterior covariance conditions. Consequently, there exists an estimation strategy providing an estimate $\hat{\mathbf z}_t= \mathbb{E}[\mathbf z_{t} | \mathbf y^{t}]$ such that the error in the unstable state vanishes asymptotically
    \begin{equation} \label{eq:error_convergence}
        \lim_{t \to \infty} \mathbb{E}[\|\boldsymbol\epsilon_t\|^2] = 0,
        \qquad
        \boldsymbol\epsilon_t := \mathbf z_t - \hat{\mathbf z}_t.
    \end{equation}

    Since the pair $(\mathbf A, \mathbf B)$ is stabilizable, the unstable subsystem $(\mathbf A_u, \mathbf B_u)$ is stabilizable. Thus, there exists a feedback gain matrix $\mathbf K$ such that $\mathbf A_{\mathrm{cl}} := \mathbf A_u + \mathbf B_u \mathbf K$ is Schur.
    We use the certainty-equivalence control law
    \begin{equation*}
        \mathbf u_t = \mathbf K\hat{\mathbf z}_t.
    \end{equation*}
    Substituting $\hat{\mathbf z}_t = \mathbf z_t - \boldsymbol\epsilon_t$ into the unstable dynamics gives
    \begin{align} \label{eq:closed_loop}
        \mathbf z_{t+1} &= \mathbf A_u \mathbf z_t + \mathbf B_u \mathbf u_t \notag \\
        &= \mathbf A_u \mathbf z_t + \mathbf B_u \mathbf K (\mathbf z_t - \boldsymbol\epsilon_t) \notag \\
        &= (\mathbf A_u + \mathbf B_u \mathbf K) \mathbf z_t - \mathbf B_u \mathbf K \boldsymbol\epsilon_t \notag \\
        &= \mathbf A_{\mathrm{cl}} \mathbf z_t + \mathbf d_t,
    \end{align}
    where
    $\mathbf d_t:=-\mathbf B_u\mathbf K\boldsymbol\epsilon_t$
    is the input induced by the estimation error.

    The closed-loop dynamics \eqref{eq:closed_loop} represent a stable linear system driven by the input $\mathbf d_t$. Since $\mathbf A_{\mathrm{cl}}$ is Schur stable, the system is Input-to-State Stable (ISS).

    From \eqref{eq:error_convergence}, the driving term converges to zero in the mean-square sense
    \begin{equation*}
        \lim_{t \to \infty} \mathbb{E}[\|\mathbf d_t\|^2]
        \le \|\mathbf B_u \mathbf K\|^2
        \lim_{t \to \infty} \mathbb{E}[\|\boldsymbol\epsilon_t\|^2]
        = 0.
    \end{equation*}
    For a linear system with a stable system matrix, if the input converges to zero, the state also converges to zero
    \begin{equation*}
        \lim_{t \to \infty} \mathbb{E}[\|\mathbf z_t\|^2] = 0.
    \end{equation*}
    It follows that $\mathbb{E}[\|\mathbf u_t\|^2]\to0$. The stable subsystem satisfies
    \begin{equation*}
        \mathbf s_{t+1}=\mathbf A_s\mathbf s_t+\mathbf B_s\mathbf u_t,
    \end{equation*}
    with $\mathbf A_s$ Schur. Hence $\mathbb{E}[\|\mathbf s_t\|^2]\to0$ as well. Since
    $\mathbf x_t=\mathbf T^{-1}[\mathbf z_t^T,\mathbf s_t^T]^T$, the full state satisfies
    \begin{equation*}
        \lim_{t\to\infty}\mathbb{E}[\|\mathbf x_t\|^2]=0.
    \end{equation*}

    This proves the attractivity condition.
\end{IEEEproof}

\begin{remark}
Theorems~\ref{thm-Sufficiency-Observability}
and~\ref{thm-Sufficiency-Stabilizability} establish the attractivity
components of Definitions~\ref{def:asymp_ms_obs}
and~\ref{def:asymp_ms_stab}. Let $\boldsymbol{\xi}_t$ denote either
the estimation error $\mathbf e_t$ or the closed-loop state $\mathbf x_t$.
If the corresponding estimator or controller design admits constants $r_\xi>0$ and a
continuous increasing function $\alpha_\xi$ with $\alpha_\xi(0)=0$ such
that
\begin{equation*}
\sup_{t\ge0}\mathbb E[\|\boldsymbol{\xi}_t\|^2]
\le
\alpha_\xi\!\left(
\mathbb E[\|\boldsymbol{\xi}_0\|^2]
\right)
\end{equation*}
whenever $\mathbb E[\|\boldsymbol{\xi}_0\|^2]\le r_\xi$, then the
associated stability component follows by choosing
$\delta\le r_\xi$ such that $\alpha_\xi(\delta)\le\epsilon$. Equivalently,
for every $\epsilon>0$, sufficiently small
$\mathbb E[\|\boldsymbol{\xi}_0\|^2]\le\delta$ guarantees
$\mathbb E[\|\boldsymbol{\xi}_t\|^2]\le\epsilon$ for every $t\ge0$.
Together with
$\lim_{t\to\infty}\mathbb E[\|\boldsymbol{\xi}_t\|^2]=0$, this yields
asymptotic mean-square observability or stabilizability, respectively. The
condition on directed information controls the attractivity that depends on
sensing, whereas the overshoot bound depends on the estimator or controller
realization.
\end{remark}

\subsection{A Lower Bound Based on Entropy Power for Nonlinear Sensing}
\label{subsec:epi_lower_bound}

The directed information rate in the sufficiency condition is generally
difficult to evaluate for nonlinear observations. When an
observation is obtained by applying a nonlinear function to the unstable
state and then adding independent noise, the entropy power inequality
\cite{Bobkov2017EPI} gives a lower bound
that depends only on the conditional entropy power of the noiseless sensor
output and the entropy power of the observation noise.

For a $p$-dimensional random vector $\mathbf a$ with finite differential
entropy, define its entropy power and conditional entropy power,
respectively, as
\begin{equation*}
    \mathrm{EP}(\mathbf a)
    :=\frac{1}{2\pi e}2^{\frac{2}{p}h(\mathbf a)},
    \qquad
    \mathrm{EP}(\mathbf a|\mathbf b)
    :=\frac{1}{2\pi e}2^{\frac{2}{p}h(\mathbf a|\mathbf b)}.
\end{equation*}

Specifically, consider observations generated according to
\begin{equation*}
    \mathbf y_t
    =
    g_t(\mathbf z_t)+\mathbf v_t,
    \qquad t\in\mathbb N_0,
\end{equation*}
where $g_t:\mathbb R^{n_u}\to\mathbb R^p$ is deterministic,
and denote the noiseless sensor output by
\begin{equation*}
    \mathbf r_t:=g_t(\mathbf z_t).
\end{equation*}
Suppose, for every $t$, that $\mathbf v_t$ is independent of
$(\mathbf z_t,\mathbf y^{t-1})$ and that $\mathbf v_t$ and
$\mathbf r_t|\mathbf y^{t-1}$ have absolutely continuous
distributions with finite differential entropies.
Whenever the following terms are well defined, define the associated
asymptotic lower bound as
\begin{equation*}
\label{eq:epi_lower_bound_rate_definition}
\begin{aligned}
\mathcal I_{\min}
:={}
\liminf_{T\to\infty}\frac{1}{T+1}
\sum_{t=0}^{T}
\frac{p}{2}
\log_2\!\left(
1+
\frac{
\mathrm{EP}(\mathbf r_t|\mathbf y^{t-1})
}{
\mathrm{EP}(\mathbf v_t)
}
\right).
\end{aligned}
\end{equation*}

\begin{proposition}[Lower Bound Based on Entropy Power]
\label{prop:epi_directed_info_lower_bound}
Under the preceding conditions,
\begin{equation}
\label{eq:epi_directed_info_lower_bound}
    \mathcal I(\mathbf z\to\mathbf y)
    \ge
    \mathcal I_{\min}.
\end{equation}
Moreover, if $\mathbf v_t\sim\mathcal N(\mathbf 0,\mathbf R_t)$ with
$\mathbf R_t\succ\mathbf 0$, then
\begin{equation*}
    \mathrm{EP}(\mathbf v_t)
    =
    (\det\mathbf R_t)^{1/p}.
\end{equation*}
\end{proposition}

\begin{IEEEproof}
See \apprefappendix{app:proof_epi_directed_info_lower_bound}{D}.
\end{IEEEproof}

For a specialization with bounds uniform over time, suppose additionally that there
exist constants $E_r>0$ and $E_v<\infty$ such that
\begin{equation}
\label{eq:uniform_entropy_power_bounds}
    \mathrm{EP}(\mathbf r_t|\mathbf y^{t-1})
    \ge E_r,
    \qquad
    \mathrm{EP}(\mathbf v_t)
    \le E_v,
    \quad \forall t.
\end{equation}
Then every term in
\eqref{eq:epi_lower_bound_rate_definition} has the uniform lower bound
\begin{equation*}
    \mathcal I(\mathbf z\to\mathbf y)
    \ge
    \mathcal I_{\min}
    \ge
    \frac p2\log_2\!\left(
        1+\frac{E_r}{E_v}
    \right),
\end{equation*}
whose right-hand side is independent of time.

Consequently, under Assumption~\ref{assump:cond_number}, the computable
condition
\begin{equation*}
    \mathcal I_{\min}
    >
    R_{\mathrm{exp}}+R_{\mathrm{NG}}
\end{equation*}
implies the conclusion of
Theorem~\ref{thm-Sufficiency-Observability}. If
$(\mathbf A,\mathbf B)$ is also stabilizable, it implies the
conclusion of Theorem~\ref{thm-Sufficiency-Stabilizability}. Under
\eqref{eq:uniform_entropy_power_bounds}, the following condition, which is independent of time,
\begin{equation*}
    \frac p2\log_2\!\left(
        1+\frac{E_r}{E_v}
    \right)
    >R_{\mathrm{exp}}+R_{\mathrm{NG}}
\end{equation*}
is therefore sufficient. The bound
reduces the evaluation of directed information to the conditional entropy power
of the noiseless sensor output; obtaining that entropy power remains
dependent on the model.

\subsection{Zero Posterior Non-Gaussianity Rate}

The general sufficiency results allow any finite $R_{\mathrm{NG}}$. This
subsection focuses on the important class $R_{\mathrm{NG}}=0$, for which the
threshold based on directed information reduces to $R_{\mathrm{exp}}$. We first show
that linear Gaussian observations have $R_{\mathrm{NG}}=0$, a
property preserved by invertible readouts, and directly verify
Assumption~\ref{assump:cond_number} in the scalar case. We then give a
criterion based on curvature for establishing $R_{\mathrm{NG}}=0$ without
explicitly computing the posterior.

\subsubsection{Linear Gaussian Observations}
\label{subsubsec:linear_gaussian_sensing}

Consider the linear Gaussian observation model given by
\begin{equation*}
    \mathbf y_k
    =
    \mathbf H_k\mathbf z_k+\mathbf v_k,
    \qquad
    \mathbf v_k\sim
    \mathcal N(\mathbf 0,\mathbf R_k),
\end{equation*}
where $\mathbf z_0$ is Gaussian,
$\mathbf R_k\succ\mathbf 0$, and the observation noises are
mutually independent and independent of the initial state. This is the
standard model underlying Kalman filtering
\cite{anderson2005optimal}.

For each fixed observation history, the corresponding control sequence is
known. The linear dynamics and Gaussian likelihood therefore preserve
Gaussianity under Bayesian updating, giving
\begin{equation*}
    p(\mathbf z_t|\mathbf y^{t})
    =
    \mathcal N\!\left(
        \hat{\mathbf z}_t,
        \mathbf P_t
    \right).
\end{equation*}
Hence, for every $t$ and almost every observation history,
\begin{equation*}
    p_t=q_t,
    \qquad
    D_t=0,
    \qquad
    R_{\mathrm{NG}}=0.
\end{equation*}

For the scalar specialization
\begin{equation*}
    y_k
    =
    h_k z_k+v_k,
    \qquad
    v_k\sim\mathcal N(0,\sigma^2),
\end{equation*}
suppose that $|h_k|\ge h_{\min}>0$. Its posterior variance is
deterministic. As derived in
\apprefappendix{app:scalar_gaussian_covariance_update}{B}, the Gaussian
update gives
\begin{align*}
    0<P_t
    &=
    \left(
        (P_t^-)^{-1}
        +\frac{h_t^2}{\sigma^2}
    \right)^{-1}
    \le\frac{\sigma^2}{h_{\min}^2},\\
    \frac{\lambda_{\max}(P_t)}
         {\lambda_{\min}(P_t)}
    &=1,
    \qquad
    \mathrm{Var}\!\left[
        \log_2P_t
    \right]=0.
\end{align*}
Since $\mathrm{Tr}(P_t)=P_t$,
Assumption~\ref{assump:cond_number} also holds.

\noindent\textit{Invertible readout extension.}
Let
\begin{equation*}
    \widetilde{\mathbf y}_k
    =
    \mathbf H_k\mathbf z_k+\mathbf v_k,
    \qquad
    \mathbf y_k
    =
    \varphi_k(\widetilde{\mathbf y}_k),
\end{equation*}
where $\varphi_k$ is a known invertible mapping. Since
$\mathbf y^{t}$ and $\widetilde{\mathbf y}^{t}$ determine each other,
\begin{equation*}
    p(\mathbf z_t|\mathbf y^{t})
    =
    p(\mathbf z_t|\widetilde{\mathbf y}^{t}).
\end{equation*}
Thus, the nonlinear readout preserves $R_{\mathrm{NG}}$ and all posterior
covariance conditions. In particular, it has $R_{\mathrm{NG}}=0$ and
inherits Assumption~\ref{assump:cond_number} whenever the latent
linear Gaussian model satisfies it. Calibrated camera response functions
over nonsaturated operating ranges provide one example
\cite{tai2013nonlinear}.

\subsubsection{Curvature Conditions for Zero Posterior Non-Gaussianity Rate}

The curvature conditions below are not assumptions of the general
sufficiency theorems. For a general nonlinear and non-Gaussian observation
model, directly evaluating $R_{\mathrm{NG}}$ requires evaluating the
posterior and its expected divergence from the moment-matched Gaussian,
which are typically unavailable in closed form. We therefore seek conditions stated directly in
terms of the observation likelihood and the initial prior. The conditions
below imply eventual posterior log-concavity, providing a verifiable route to
$R_{\mathrm{NG}}=0$. This route verifies only the property $R_{\mathrm{NG}}=0$;
Assumption~\ref{assump:cond_number} remains a separate covariance regularity
requirement. When both hold, the general threshold reduces to
$R_{\mathrm{exp}}$.

Following the observed information convention in~\cite{demaio2018observedFIM}, define
the local observed Fisher information matrix (OFIM) of the unstable
state at time $k$ as
\begin{equation*}
    \mathbf F_k
    :=
    -\nabla_{\mathbf z_k}^2
    \log p(\mathbf y_k|\mathbf z_k).
\end{equation*}
For $k\le t$, iterating the unstable dynamics from time $k$ to time $t$
gives
\begin{equation*}
    \mathbf z_t
    =
    \mathbf A_u^{t-k}\mathbf z_k
    +
    \sum_{\ell=k}^{t-1}
    \mathbf A_u^{t-1-\ell}
    \mathbf B_u\mathbf u_\ell.
\end{equation*}
Since $\mathbf A_u$ is invertible, the corresponding inverse dynamics is
\begin{equation*}
    \mathbf z_k
    =
    \mathbf A_u^{-(t-k)}\mathbf z_t
    +\mathbf c_{k,t},
\end{equation*}
where
\begin{equation*}
    \mathbf c_{k,t}
    :=
    -\sum_{\ell=k}^{t-1}
    \mathbf A_u^{-(\ell-k+1)}
    \mathbf B_u\mathbf u_\ell.
\end{equation*}
Conditioned on the realized
control history, $\mathbf c_{k,t}$ is fixed and does
not depend on the terminal state argument $\mathbf z_t$. For any integer
$L\in[1:t+1]$, define the
cumulative window OFIM in the coordinates of $\mathbf z_t$ by
\begin{equation*}
    \mathbf F_{t,L}
    :=
    \sum_{k=t-L+1}^{t}
    (\mathbf A_u^{-(t-k)})^T
    \mathbf F_k
    \mathbf A_u^{-(t-k)}.
\end{equation*}
For notational simplicity, the dependence of
$\mathbf F_k$ and
$\mathbf F_{t,L}$ on the relevant state
arguments, observation realizations, and realized control history is
suppressed. The cumulative window OFIM is left unnormalized because
division by the fixed window length $L$ would only rescale the constant
below.

\begin{assumption}
\label{assump:obs_hessian}
For the verification based on curvature, the observation process is
assumed to be conditionally memoryless with respect to the unstable
state, i.e.,
\begin{equation*}
    p\!\left(
        \mathbf y_k
        |
        \mathbf z^{k},\mathbf y^{k-1}
    \right)
    =
    p(\mathbf y_k|\mathbf z_k),
    \qquad k\in\mathbb{N}_0.
\end{equation*}
The corresponding local observation log-likelihoods are assumed to be
twice continuously differentiable with respect to the unstable state
argument. There exist an integer $L\ge1$ and constants $\alpha>0$
and $0\le C_H<\infty$ such that
\begin{equation*}
    \lambda_{\min}\!\left(
        \mathbf F_{t,L}
    \right)
    \ge \alpha,
    \qquad \forall t\ge L-1,
\end{equation*}
and
\begin{equation*}
    \lambda_{\min}\!\left(
        \mathbf F_k
    \right)
    \ge -C_H,
    \qquad \forall k\in[0\!:\!L-2].
\end{equation*}
The second condition is vacuous when $L=1$. Both inequalities hold for
all relevant arguments of the unstable state and almost every observation
realization.
\end{assumption}

Assumption~\ref{assump:obs_hessian} requires every length-$L$ window to
provide at least $\alpha$ units of observed information in every unstable
direction. The second bound prevents the finitely many terms in the
initial incomplete block from having arbitrarily negative observed
information.

\begin{assumption} \label{assump:prior_hessian}
The initial unstable state $\mathbf z_0$ has a strictly positive PDF
$p(\mathbf z_0)$ whose log-density is twice continuously
differentiable. There exists a scalar $\beta > 0$ such that for all
$\mathbf z_0 \in \mathbb{R}^{n_u}$
\begin{equation*}
    \mathbf H_0
    :=
    \nabla_{\mathbf z_0}^2 \log p(\mathbf z_0),
    \qquad
    \mathbf H_0\preceq \beta \mathbf I_{n_u}.
\end{equation*}
\end{assumption}

Only an upper Hessian bound is imposed, so the prior need not be
log-concave.

\begin{lemma}
\label{lemma:log_concavity}
Consider the unstable subsystem. Under
Assumptions~\ref{assump:obs_hessian}
and~\ref{assump:prior_hessian}, there exist a time $T_{lc}$ and a
constant $c>0$ such that, for every $t>T_{lc}$, the posterior density
$p_t$ satisfies, almost surely,
\begin{equation*}
    \nabla_{\mathbf z}^2
    \log p_t(\mathbf z)
    \preceq
    -c\mathbf I_{n_u},
    \qquad
    \forall\,\mathbf z\in\mathbb R^{n_u}.
\end{equation*}
\end{lemma}

\begin{IEEEproof}
    See \apprefappendix{app:proof_lemma_log_concavity}{C}.
\end{IEEEproof}

\begin{lemma}[Entropy Gap for a Log-Concave Posterior]
\label{lemma:entropy_gap}
For any $t$ and fixed observation history $\mathbf y^{t}$, suppose
that $p_t$ is an absolutely continuous log-concave density
on $\mathbb R^{n_u}$. Let
$q_t$ be its moment-matched Gaussian density. Then
\begin{equation*}
0
\le
D_t
=
h\!\left(q_t\right)
-h\!\left(p_t\right)
\le
n_u C_{n_u},
\end{equation*}
where $C_{n_u}<\infty$ depends only on the dimension $n_u$.
\end{lemma}

\begin{IEEEproof}
The bounds on the entropy gap for
log-concave densities follow from \cite{Bobkov2011,Bobkov2012} and
\cite[Corollary 3]{Bobkov2010ISIT}.
\end{IEEEproof}

\begin{proposition}[Curvature Conditions for $R_{\mathrm{NG}}=0$]
\label{prop:curvature_implies_entropy_deficit}
Under Assumptions~\ref{assump:obs_hessian}
and~\ref{assump:prior_hessian},
the posterior non-Gaussianity rate satisfies
\begin{equation*}
    R_{\mathrm{NG}}=0.
\end{equation*}
\end{proposition}

\begin{IEEEproof}
By Lemma~\ref{lemma:log_concavity}, there exists a finite time
$T_{lc}$ such that $p_t$ is log-concave almost surely for
every $t>T_{lc}$. Applying Lemma~\ref{lemma:entropy_gap} to the posterior
and its moment-matched Gaussian gives
\begin{equation}
\label{eq:uniform_posterior_gaussian_divergence}
0
\le
D_t
\le n_u C_{n_u}
\end{equation}
for every $t>T_{lc}$, almost surely.
Taking expectation, dividing by $t+1$, and using
\eqref{eq:uniform_posterior_gaussian_divergence} yields
\begin{align*}
0
&\le
\frac{1}{t+1}
\mathbb{E}[D_t]
\le
\frac{n_u C_{n_u}}{t+1}.
\end{align*}
The right-hand side tends to zero, so
Definition~\ref{def:posterior_non_gaussianity_rate} gives
$R_{\mathrm{NG}}=0$.
\end{IEEEproof}

\begin{corollary}[Sufficiency under Curvature Conditions]
\label{cor:curvature_based_sufficiency}
Under Assumptions~\ref{assump:cond_number},
\ref{assump:obs_hessian}, and~\ref{assump:prior_hessian}, if
$\mathcal{I}(\mathbf z\to\mathbf y)>R_{\mathrm{exp}}$, then the
conclusion of Theorem~\ref{thm-Sufficiency-Observability} holds. If, in
addition, $(\mathbf A,\mathbf B)$ is stabilizable, then the
conclusion of Theorem~\ref{thm-Sufficiency-Stabilizability} also holds.
\end{corollary}

\section{Conclusion} \label{sec:conclusion}

In this paper, we considered estimation and feedback control of an unstable linear system whose measurements are produced by a prescribed, possibly nonlinear and non-Gaussian, observation law.
Classical data-rate results place the information bottleneck after measurement, at the communication channel. 
In contrast, our formulation places it at the sensing stage and asks whether the observation mechanism itself preserves enough sensing information about unstable modes.
Using directed information, we showed that the open-loop expansion rate remains the fundamental information demand for mean-square observability and stabilizability, including under additive process disturbances.
We further translated this information requirement into practical criteria: an upper bound certifies infeasibility, while a lower bound certifies sufficiency under posterior covariance regularity. 
For linear Gaussian sensing, the information rate is characterized exactly by the steady-state Riccati equation.
For general posteriors, our analysis showed that posterior non-Gaussianity can prevent a reduction in posterior entropy from translating directly into mean-square convergence.
This effect is quantified by the posterior non-Gaussianity rate, which vanishes under verifiable curvature conditions.
Overall, the results separate the sensing requirement into the information needed to counter unstable dynamics and the extra cost caused by posterior geometry, providing a unified basis for judging whether a given observation mechanism is adequate for estimation and control.

\ifarxivversion
\appendices
\section{Proof of Theorem~\ref{thm-Necessary-Stabilizability}} \label{app:proof_thm_stab}

\begin{IEEEproof}
The proof follows a similar logic to Theorem~\ref{thm-Necessary-Observability}, utilizing the entropy balance equation derived in \eqref{eq:rate_balance_final}.
As established in the proof of Theorem~\ref{thm-Necessary-Observability}, the entropy of the unstable subsystem satisfies the balance equation
\begin{equation}
    \frac{1}{t+1} I(\mathbf z^{t} \to \mathbf y^{t}) = R_{\mathrm{exp}} + \frac{h_0}{t+1} - \frac{h_{t+1}}{t+1},
    \label{eq:rate_balance_stab_app}
\end{equation}
where $h_{t+1} := h(\mathbf z_{t+1} | \mathbf y^{t})$ is the terminal conditional differential entropy.

We focus on showing that
\begin{equation*}
    \limsup_{t\to\infty}\frac{h_{t+1}}{t+1}\le 0
\end{equation*}
under the stabilizability hypothesis.
By mean-square stabilizability,
\begin{equation*}
    \limsup_{t\to\infty}
    \mathbb E[\|\mathbf x_t\|^2]<\infty.
\end{equation*}
Hence, there exist constants $C_x<\infty$ and $t_0$ such that
$\mathbb E[\|\mathbf x_t\|^2]\le C_x$ for all $t\ge t_0$.
Because $\mathbf z_t$ is the unstable block of $\mathbf T\mathbf x_t$,
for all sufficiently large $t$,
\begin{equation*}
    \mathbb E[\|\mathbf z_{t+1}\|^2]
    \le
    \|\mathbf T\|^2
    \mathbb E[\|\mathbf x_{t+1}\|^2]
    \le
    \|\mathbf T\|^2 C_x
    =: C_z.
\end{equation*}
Let
\begin{equation*}
    \boldsymbol\Sigma_{z,t+1}
    :=
    \operatorname{Cov}(\mathbf z_{t+1}).
\end{equation*}
Then
\begin{equation*}
    \operatorname{Tr}(\boldsymbol\Sigma_{z,t+1})
    =
    \mathbb E\!\left[
        \|\mathbf z_{t+1}
        -\mathbb E[\mathbf z_{t+1}]\|^2
    \right]
    \le
    \mathbb E[\|\mathbf z_{t+1}\|^2]
    \le C_z.
\end{equation*}
Therefore,
\begin{align*}
    h_{t+1}
    &= h(\mathbf z_{t+1}\mid\mathbf y^t)\\
    &\le h(\mathbf z_{t+1})\\
    &\le
    \frac{n_u}{2}
    \log_2\!\left(
        \frac{2\pi e}{n_u}
        \operatorname{Tr}(\boldsymbol\Sigma_{z,t+1})
    \right)\\
    &\le
    \frac{n_u}{2}
    \log_2\!\left(
        \frac{2\pi e}{n_u}C_z
    \right)
    =:H_{\mathrm{bound}}
\end{align*}
for all sufficiently large $t$.
Consequently,
\begin{equation*}
    \limsup_{t\to\infty}
    \frac{h_{t+1}}{t+1}
    \le 0.
\end{equation*}

Since the initial entropy $h_0$ is finite, taking the limit inferior of \eqref{eq:rate_balance_stab_app} gives
\begin{equation*}
    \liminf_{t \to \infty} \frac{1}{t+1} I(\mathbf z^{t} \to \mathbf y^{t})
    \ge R_{\mathrm{exp}}.
\end{equation*}
This completes the proof.
\end{IEEEproof}

\section{Derivation of the Scalar Gaussian Posterior Variance}
\label{app:scalar_gaussian_covariance_update}

Conditioned on the past observation history, the scalar prediction density satisfies
\begin{equation*}
    z_t\mid\mathbf y^{t-1}
    \sim
    \mathcal N\!\left(
        \hat z_t^-,
        P_t^-
    \right).
\end{equation*}
Hence,
\begin{align*}
p(z_t\mid\mathbf y^{t-1})=
\frac{1}{\sqrt{2\pi P_t^-}}
\exp\!\left[
-\frac{(z_t-\hat z_t^-)^2}
      {2P_t^-}
\right].
\end{align*}
For the observation
\begin{equation*}
    y_t=h_t z_t+v_t,
    \qquad
    v_t\sim\mathcal N(0,\sigma^2),
\end{equation*}
the likelihood is
\begin{equation*}
    p(y_t\mid z_t)
    =
    \frac{1}{\sqrt{2\pi\sigma^2}}
    \exp\!\left[
        -\frac{(y_t-h_t z_t)^2}{2\sigma^2}
    \right].
\end{equation*}
Bayes' rule therefore gives
\begin{align*}
p(z_t\mid\mathbf y^{t})=
\frac{
    p(y_t\mid z_t)
    p(z_t\mid\mathbf y^{t-1})
}{
    p(y_t\mid\mathbf y^{t-1})
}.
\end{align*}
Substituting the prediction density and the likelihood into Bayes' rule gives
\begin{align*}
p(z_t\mid\mathbf y^{t})\propto
\exp\!\left\{
-\frac12\left[
\frac{(z_t-\hat z_t^-)^2}
     {P_t^-}
+\frac{(y_t-h_t z_t)^2}{\sigma^2}
\right]
\right\},
\end{align*}
where the last proportionality retains all terms that depend on $z_t$.
Define
\begin{align*}
A_t
&:=
(P_t^-)^{-1}
+\frac{h_t^2}{\sigma^2},\\
B_t
&:=
\frac{\hat z_t^-}
     {P_t^-}
+\frac{h_t y_t}{\sigma^2},\\
C_t
&:=
\frac{(\hat z_t^-)^2}
     {P_t^-}
+\frac{y_t^2}{\sigma^2}.
\end{align*}
Expanding both squared terms gives
\begin{align*}
\frac{(z_t-\hat z_t^-)^2}
      {P_t^-}
+\frac{(y_t-h_t z_t)^2}{\sigma^2}
&=
A_tz_t^2-2B_tz_t+C_t\\
&=
A_t\left(z_t-\frac{B_t}{A_t}\right)^2
+C_t-\frac{B_t^2}{A_t}.
\end{align*}
The last two terms not involving $z_t$ are absorbed into the
normalizing constant. The posterior is therefore
\begin{equation*}
    z_t\mid\mathbf y^{t}
    \sim
    \mathcal N\!\left(
        \frac{B_t}{A_t},
        A_t^{-1}
    \right).
\end{equation*}
In particular, its variance is
\begin{equation*}
    P_t
    = (A_t)^{-1}
    =
    \left(
        (P_t^-)^{-1}
        +\frac{h_t^2}{\sigma^2}
    \right)^{-1}.
\end{equation*}
Multiplying the numerator and denominator by
$\sigma^2P_t^-$ gives the equivalent covariance update
\begin{equation*}
    P_t
    =
    \frac{
        \sigma^2P_t^-
    }{
        \sigma^2+h_t^2P_t^-
    }.
\end{equation*}
Since
\begin{equation*}
    (P_t^-)^{-1}
    +\frac{h_t^2}{\sigma^2}
    \ge
    \frac{h_t^2}{\sigma^2}>0,
\end{equation*}
the posterior variance satisfies
\begin{align*}
    0<P_t
    \le
    \left(\frac{h_t^2}{\sigma^2}\right)^{-1}
    =
    \frac{\sigma^2}{h_t^2}
    \le
    \frac{\sigma^2}{h_{\min}^2}.
\end{align*}

\section{Proof of Lemma~\ref{lemma:log_concavity}}\label{app:proof_lemma_log_concavity}
\begin{IEEEproof}
We first derive the posterior density in a form that explicitly accounts
for causal feedback. Fix an observation history $\mathbf y^{t}$.
Because the control policy is causal, the
corresponding control realization $\mathbf u^{t-1}$ is also fixed.
Since the unstable dynamics are noiseless and $\mathbf A_u$ is
invertible, every terminal state argument $\mathbf z_t$ determines
the earlier state arguments as
\begin{equation*}
    \mathbf z_k
    =
    \mathbf A_u^{-(t-k)}\mathbf z_t
    +\mathbf c_{k,t},
    \qquad k\le t,
\end{equation*}
and the Jacobian of this mapping with respect to $\mathbf z_t$ is
$\mathbf A_u^{-(t-k)}$.

By the probability chain rule and the conditional memoryless property in
Assumption~\ref{assump:obs_hessian},
\begin{align*}
p_{\mathbf y^{t}|\mathbf z_0}
(\mathbf y^{t}|\mathbf z_0)=
\prod_{k=0}^{t}
p_{\mathbf y_k|\mathbf z_k}
(\mathbf y_k|\mathbf z_k).
\end{align*}
For the fixed observation history, the change of variables
$\mathbf z_0=\mathbf A_u^{-t}\mathbf z_t
+\mathbf c_{0,t}$ therefore gives
\begin{align*}
p_t(\mathbf z_t)=
\frac{|\det(\mathbf A_u^{-t})|}
     {p_{\mathbf y^{t}}(\mathbf y^{t})}
p_{\mathbf z_0}\!\left(
    \mathbf A_u^{-t}\mathbf z_t
    +\mathbf c_{0,t}
\right)
\prod_{k=0}^{t}
p_{\mathbf y_k|\mathbf z_k}
(\mathbf y_k|\mathbf z_k).
\end{align*}
Taking logarithms, all terms contributed by the Jacobian and the
normalizing density $p_{\mathbf y^{t}}(\mathbf y^{t})$ are
constant with respect to $\mathbf z_t$. Recall the prior Hessian
\begin{equation*}
    \mathbf H_0
    :=
    \nabla_{\mathbf z_0}^2
    \log p_{\mathbf z_0}(\mathbf z_0).
\end{equation*}
Applying the Hessian chain rule to the transformed prior and each local
observation log-likelihood gives the total posterior Hessian
$\mathbf H_t:=\nabla_{\mathbf z_t}^2
\log p_t(\mathbf z_t)$ as
\begin{align*}
    \mathbf H_t
    &= (\mathbf A_u^{-t})^T \mathbf H_0 (\mathbf A_u^{-t})
    - \sum_{k=0}^t (\mathbf A_u^{-(t-k)})^T \mathbf F_k (\mathbf A_u^{-(t-k)}) \\
    &\stackrel{(a)}{=} (\mathbf A_u^{-t})^T \mathbf H_0 (\mathbf A_u^{-t}) \\
    &\quad - \sum_{j=0}^{N_t-1} \left( \sum_{k \in \mathcal{T}_j} (\mathbf A_u^{-(t-k)})^T \mathbf F_k (\mathbf A_u^{-(t-k)}) \right) + \mathbf R_t \\
    &\stackrel{(b)}{\preceq} \beta (\mathbf A_u^{-t})^T (\mathbf A_u^{-t}) - \alpha \sum_{j=0}^{N_t-1} (\mathbf A_u^{-jL})^T (\mathbf A_u^{-jL}) + \mathbf R_t,
\end{align*}
where
\begin{itemize}
    \item[(a)] We partition the time horizon $[0\!:\!t]$ into
    $N_t = \lfloor (t+1)/L \rfloor$ complete blocks of length $L$,
    indexed by $j \in [0:N_t-1]$, where
    \begin{equation*}
        \mathcal{T}_j
        :=
        [t-(j+1)L+1:t-jL].
    \end{equation*}
    The remaining initial indices form
    $\mathcal{T}_{\mathrm{rem}}=[0:t-N_tL]$, which is understood to be
    empty when $t-N_tL=-1$. The corresponding remainder term is
    \begin{equation*}
        \mathbf R_t := -\sum_{k=0}^{t - N_t L} (\mathbf A_u^{-(t-k)})^T \mathbf F_k (\mathbf A_u^{-(t-k)}).
    \end{equation*}
    \item[(b)] Assumption~\ref{assump:prior_hessian} gives
    $\mathbf H_0\preceq\beta\mathbf I_{n_u}$. Moreover,
    by the definition of the cumulative window OFIM, each complete block
    satisfies
    \begin{align*}
    &-\sum_{k\in\mathcal{T}_j}
    (\mathbf A_u^{-(t-k)})^T
    \mathbf F_k
    \mathbf A_u^{-(t-k)}
    \\
    &\quad=
    -(\mathbf A_u^{-jL})^T
    \mathbf F_{t-jL,L}
    \mathbf A_u^{-jL}
    \\
    &\quad\preceq
    -\alpha
    (\mathbf A_u^{-jL})^T
    \mathbf A_u^{-jL},
    \end{align*}
    where the last inequality follows from
    Assumption~\ref{assump:obs_hessian}.
\end{itemize}

To analyze the asymptotic spectral behavior, we use the real Jordan
decomposition $\mathbf A_u^{-1} = \mathbf V \mathbf J \mathbf V^{-1}$, where the Jordan matrix $\mathbf J$ has the block form
\begin{equation*}
    \mathbf J = \begin{bmatrix} \mathbf J_{<1} & \mathbf 0 \\ \mathbf 0 & \mathbf J_{=1} \end{bmatrix}.
\end{equation*}
Here, $\mathbf J_{<1}$ corresponds to the eigenvalues with magnitude strictly less than 1, and $\mathbf J_{=1}$ corresponds to the eigenvalues with magnitude equal to 1. Let $n_{<1}:=\dim(\mathcal S_{<1})$ denote the dimension of the strictly unstable subspace.

First, we bound the remainder terms. The second OFIM bound in
Assumption~\ref{assump:obs_hessian} gives
$\mathbf F_k\succeq-C_H\mathbf I_{n_u}$ and hence
\begin{equation*}
    -\mathbf V^T \mathbf F_k \mathbf V
    \preceq C_H\mathbf V^T\mathbf V
    \preceq C_{\text{rem}}\mathbf I_{n_u},
\end{equation*}
where $C_{\text{rem}}:=C_H\sigma_{\max}^2(\mathbf V)$.
The identity
$\mathbf A_u^{-k}\mathbf V=\mathbf V\mathbf J^k$
and the bounds on the singular values
\begin{equation*}
    \sigma_{\min}^2(\mathbf V)\mathbf I_{n_u}
    \preceq
    \mathbf V^T\mathbf V
    \preceq
    \sigma_{\max}^2(\mathbf V)\mathbf I_{n_u}
\end{equation*}
give
\begin{align*}
&\mathbf V^T(\mathbf A_u^{-t})^T
(\mathbf A_u^{-t})\mathbf V
\preceq
\sigma_{\max}^2(\mathbf V)
(\mathbf J^t)^T\mathbf J^t,\\
&-\mathbf V^T(\mathbf A_u^{-jL})^T
(\mathbf A_u^{-jL})\mathbf V
\preceq
-\sigma_{\min}^2(\mathbf V)
(\mathbf J^{jL})^T\mathbf J^{jL}.
\end{align*}
Combining these inequalities with the remainder bound, the transformed
Hessian
$\widetilde{\mathbf H}_t:=\mathbf V^T\mathbf H_t\mathbf V$
satisfies
\begin{align} \label{eq:transformed_H_bound}
    \widetilde{\mathbf H}_t \preceq
    &\beta \sigma_{\max}^2(\mathbf V) (\mathbf J^t)^T \mathbf J^t
    + C_{\text{rem}} \sum_{k \in \mathcal{T}_{\mathrm{rem}}} (\mathbf J^{t-k})^T (\mathbf J^{t-k}) \notag \\
    &- \alpha \sigma_{\min}^2(\mathbf V) \sum_{j=0}^{N_t-1} (\mathbf J^{jL})^T (\mathbf J^{jL}).
\end{align}

    We now evaluate the limit of \eqref{eq:transformed_H_bound} on the invariant subspaces:

    Case A: Strictly Unstable Subspace ($\mathcal{S}_{<1}$) \\
    Since the spectral radius $\rho(\mathbf J_{<1}) < 1$, terms involving powers of $\mathbf J_{<1}$ vanish as $t \to \infty$. The most recent block ($j=0$) therefore gives the dominant constant negative term
    \begin{equation*}
        \widetilde{\mathbf H}_{<1} \preceq -\alpha \sigma_{\min}^2(\mathbf V) (\mathbf J_{<1}^0)^T (\mathbf J_{<1}^0) = -\alpha \sigma_{\min}^2(\mathbf V) \mathbf I_{n_{<1}}.
    \end{equation*}
    Thus, $\widetilde{\mathbf H}_{<1}$ is asymptotically strictly negative definite.

    Case B: Marginally Unstable Subspace ($\mathcal{S}_{=1}$) \\
    Here, $\|\mathbf J_{=1}^t\| \sim O(t^{K-1})$. For any unit vector $\mathbf v$, the prior and remainder terms in \eqref{eq:transformed_H_bound} grow no faster than the largest individual term, which scales as
    \begin{equation*}
        \text{Positive} \sim \|\mathbf J_{=1}^t \mathbf v\|^2 \sim O(t^{2K-2}).
    \end{equation*}
    In contrast, the cumulative negative term scales as
    \begin{align*}
        \text{Negative}
        &\sim
        -\sum_{j=0}^{N_t-1}
        \|\mathbf J_{=1}^{jL}\mathbf v\|^2\\
        &\sim
        -\int_0^t \tau^{2K-2}\,d\tau
        \sim -O(t^{2K-1}).
    \end{align*}
    Since the order of the cumulative negative term ($2K-1$) strictly exceeds that of the positive terms ($2K-2$), $\lambda_{\max}(\widetilde{\mathbf H}_{=1}) \to -\infty$.

Consequently, $\mathbf H_t$ becomes strictly negative definite for sufficiently large $t$.
\end{IEEEproof}

\section{Proof of Proposition~\ref{prop:epi_directed_info_lower_bound}}
\label{app:proof_epi_directed_info_lower_bound}

\begin{IEEEproof}
Fix $t$ and an observation history
$\mathbf y^{t-1}=\mathbf a^{t-1}$. Since
$\mathbf r_t=g_t(\mathbf z_t)$ and
$\mathbf y_t=\mathbf r_t+\mathbf v_t$, the independence of
$\mathbf v_t$ gives
\begin{align}
I(\mathbf z_t;\mathbf y_t|
\mathbf y^{t-1}=\mathbf a^{t-1})
&=
h(\mathbf r_t+\mathbf v_t|
\mathbf y^{t-1}=\mathbf a^{t-1})
-h(\mathbf v_t).
\label{eq:app_epi_mutual_information}
\end{align}

For the fixed history, $\mathbf r_t$ and $\mathbf v_t$ are independent.
The entropy power inequality \cite{Bobkov2017EPI} therefore yields
\begin{align}
&\mathrm{EP}(\mathbf r_t+\mathbf v_t|
\mathbf y^{t-1}=\mathbf a^{t-1})
\notag\\
&\quad\ge
\mathrm{EP}(\mathbf r_t|
\mathbf y^{t-1}=\mathbf a^{t-1})
+\mathrm{EP}(\mathbf v_t).
\label{eq:app_conditional_epi}
\end{align}
Using
\begin{equation*}
    h(\mathbf a)
    =
    \frac p2\log_2\!\left(2\pi e\,\mathrm{EP}(\mathbf a)\right)
\end{equation*}
in \eqref{eq:app_epi_mutual_information} and then applying
\eqref{eq:app_conditional_epi} give
\begin{align*}
I(\mathbf z_t;\mathbf y_t|
\mathbf y^{t-1}=\mathbf a^{t-1})
\ge
\frac p2\log_2\!\left(
1+
\frac{
\mathrm{EP}(\mathbf r_t|
\mathbf y^{t-1}=\mathbf a^{t-1})
}{
\mathrm{EP}(\mathbf v_t)
}
\right).
\end{align*}

Taking expectation over the random observation history gives a valid
lower bound averaged over observation histories. Define
\begin{equation*}
f_t(a)
:=
\frac p2\log_2\!\left(
1+2^{\frac{2}{p}(a-h(\mathbf v_t))}
\right).
\end{equation*}
The function $f_t$ is convex because it is a positive scaling of a
log-sum-exp function. Jensen's inequality and the definition of conditional
differential entropy therefore imply
\begin{align}
&I(\mathbf z_t;\mathbf y_t|\mathbf y^{t-1})
\notag\\
&\quad\ge
\mathbb E_{\mathbf y^{t-1}}\!\left[
f_t\!\left(
h(\mathbf r_t|
\mathbf y^{t-1}=\mathbf a^{t-1})
\right)
\right]
\notag\\
&\quad\ge
f_t\!\left(
h(\mathbf r_t|\mathbf y^{t-1})
\right)
\notag\\
&\quad=
\frac p2\log_2\!\left(
1+
\frac{
\mathrm{EP}(\mathbf r_t|\mathbf y^{t-1})
}{
\mathrm{EP}(\mathbf v_t)
}
\right).
\label{eq:app_average_epi_bound}
\end{align}

Summing \eqref{eq:app_average_epi_bound} from $t=0$ to $t=T$ and using
the memoryless representation of directed information gives
\begin{align*}
I(\mathbf z^{T}\to\mathbf y^{T})
&=
\sum_{t=0}^{T}
I(\mathbf z_t;\mathbf y_t|\mathbf y^{t-1})\\
&\ge
\sum_{t=0}^{T}
\frac p2\log_2\!\left(
1+
\frac{
\mathrm{EP}(\mathbf r_t|\mathbf y^{t-1})
}{
\mathrm{EP}(\mathbf v_t)
}
\right).
\end{align*}
Dividing by $T+1$ and taking the limit inferior proves
\eqref{eq:epi_directed_info_lower_bound}.

Finally, if
$\mathbf v_t\sim\mathcal N(\mathbf 0,\mathbf R_t)$, then
\begin{align*}
h(\mathbf v_t)
&=
\frac12\log_2\!\left(
(2\pi e)^p\det\mathbf R_t
\right),\\
\mathrm{EP}(\mathbf v_t)
&=
\frac{1}{2\pi e}
2^{\frac{2}{p}h(\mathbf v_t)}
=
(\det\mathbf R_t)^{1/p}.
\end{align*}

\end{IEEEproof}

\fi

\bibliographystyle{IEEEtran}
\bibliography{ref}

\end{document}